\documentclass[runningheads]{llncs}
\usepackage[T1]{fontenc}
\usepackage{graphicx}
\usepackage{amsmath,amsfonts}
\usepackage{paralist}
\usepackage{enumerate}
\usepackage{color}
\usepackage{mathtools}
\usepackage{tabularx,booktabs}
\usepackage{multirow}
\usepackage{siunitx}
\usepackage{xcolor}
\usepackage{subcaption}
\usepackage{stmaryrd}
\usepackage{algorithm,algpseudocode,algorithmicx}
\algrenewcommand\algorithmicrequire{\textbf{Input:}}
\algrenewcommand\algorithmicensure{\textbf{Output:}}
\usepackage[graphicx]{realboxes}

\usepackage[utf8]{inputenc}
\usepackage{marvosym}
\usepackage[hidelinks]{hyperref}

\newcommand{\orcidlink}[1]{\href{https://orcid.org/#1}{\includegraphics[height=0.8em]{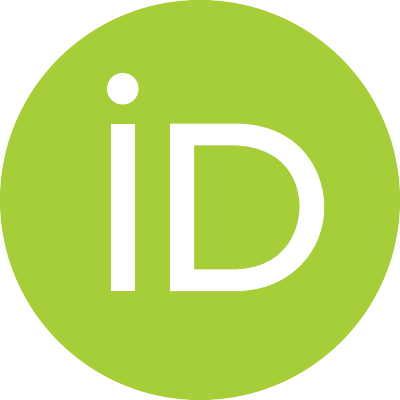}}}

\newlength{\subcolumnwidth}
\newenvironment{subcolumns}[1][0.45\columnwidth]
 {\valign\bgroup\hsize=#1\setlength{\subcolumnwidth}{\hsize}\vfil##\vfil\cr}
 {\crcr\egroup}
\newcommand{\nextsubcolumn}[1][]{%
  \cr\noalign{\hfill}
  \if\relax\detokenize{#1}\relax\else\hsize=#1\setlength{\subcolumnwidth}{\hsize}\fi
}
\newcommand{\nextsubfigure}{\vfill}

\begin{document}
\title{Supermartingale Certificates for \\ Parametric MDPs}
%
%
\author{Kaushik Mallik\inst{1}\textsuperscript{\Letter}\orcidlink{0000-0001-9864-7475} \and {\DJ}or{\dj}e \v{Z}ikeli\'{c} \inst{2}\textsuperscript{\Letter}\orcidlink{0000-0002-4681-1699}}
\authorrunning{Kaushik Mallik and {\DJ}or{\dj}e \v{Z}ikeli\'{c}}
%
\institute{IMDEA Software Institute, Madrid, Spain \\
\email{kaushik.mallik@imdea.org} \\
\and Nanyang Technological University, Singapore, Singapore \\
\email{djordje.zikelic@ntu.edu.sg}}
\maketitle              
\begin{abstract}
We consider the problems of formal verification and synthesis in parametric Markov decision processes (MDPs) with \textit{general} measurable state and action spaces.
The heart of our approach is a parameter flattening transformation, which allows us to transform parametric MDPs into semantically equivalent non-parametric MDPs.
Building on this transformation, we introduce the novel notion of \textit{parametric} supermartingale certificates, which generalize the traditional supermartingale certificates---used for non-parametric MDPs---to the parametric setting. 
We use our parametric supermartingale certificates to design algorithms for verification and approximate synthesis in polynomial arithmetic parametric MDPs.
This leads to the first verification and synthesis algorithms for parametric MDPs with \textit{general} state and action spaces. 
We implement our algorithms and experimentally evaluate them on several continuous parametric random walk benchmarks.


\end{abstract}

\newcommand{\MC}{M}
\newcommand{\MDP}{M}
\newcommand{\States}{S}
\newcommand{\StatesSA}{\mathcal{F}_S}
\newcommand{\Actions}{A}
\newcommand{\ActionsSA}{\mathcal{F}_A}
\newcommand{\strategy}{\sigma}
\newcommand{\kernel}{P}
\newcommand{\Init}{\textsf{Init}}
\newcommand{\support}{\textsf{support}}
\newcommand{\Path}{\textsf{Path}}
\newcommand{\FPath}{\textsf{FPath}}
\newcommand{\ParamSet}{U}
\newcommand{\Reach}{\textsf{Reach}}
\newcommand{\Safe}{\textsf{Safe}}
\newcommand{\StrategySpace}{\textsf{Str}}
\newcommand{\Uniform}{\textsf{Uniform}}
\newcommand{\Verify}{\mathtt{Verify}}
\newcommand{\volume}{\textsf{volume}}

\renewcommand{\diamondsuit}{\heartsuit}

\newcommand{\angelsub}{\mathtt{VerifyAngelic}}
\newcommand{\demonsub}{\mathtt{VerifyDemonic}}
\newcommand{\angelR}{\underline{R}}
\newcommand{\demonR}{\underline{R'}}
\newcommand{\sat}{\mathsf{SAT}}
\newcommand{\unsat}{\mathsf{UNSAT}}
\newcommand{\unknown}{\mathsf{UNKNOWN}}
\newcommand{\grid}[1]{\llbracket #1\rrbracket}


\section{Introduction}\label{sec:intro}

Markov decision processes (MDPs) are a standard model for sequential decision making under probabilistic uncertainty~\cite{Puterman94}. A central assumption underlying classical MDPs and their probabilistic model checking algorithms is that all transition probabilities are exactly known. In practice, however, transition probabilities often depend on quantities such as failure rates, packet loss probabilities, or environmental conditions, whose values are rarely known precisely~\cite{JansenJK22}.

Parametric Markov decision processes (pMDPs) address this limitation by allowing transition probabilities to depend on finitely many real-valued parameters~\cite{JungesAHJKQV24}. The values of parameters may be unknown or only partially known, for instance when they are constrained to lie within given bounds~\cite{GivanLD00}. As a result, pMDPs can model uncertainties in probabilistic dynamics, including complex dependencies such as multivariate polynomial transition probabilities~\cite{JansenJK22,JungesAHJKQV24}.

Substantial progress has been made in the probabilistic model checking of parametric Markov chains and finite-state pMDPs~\cite{Daws04,JansenCVWAKB14,DehnertJJCVBKA15,QuatmannD0JK16,BaierHHJKK20,HahnHZ11,Cubuktepe0JKPPT17,CubuktepeJJKT22}. The core problems studied in this setting are parameter \textit{verification}---establishing that a specification holds for all parameter valuations in a given region---and parameter \textit{synthesis}---computing regions of parameter valuations for which the specification is satisfied~\cite{JungesAHJKQV24}. Due to the inherent complexity of exact synthesis, approximate synthesis variants are also considered.

Most existing methods are restricted to pMDPs with \textit{finite} state and action spaces, whereas \textit{in}finite space models are common in many applications, such as stochastic control and robotics~\cite{fleming2012deterministic}, and probabilistic programs with unbounded or continuous state spaces~\cite{BKS2020}. Addressing this gap is our main motivation.

\smallskip\noindent{\bf Our contributions.}
We initiate the study of automated verification and synthesis for pMDPs with \textit{general} measurable state and action spaces, by using supermartingale certificates. Supermartingale certificates are locally checkable proof rules with strong theoretical foundations in probability theory~\cite{williams1991probability}, that~were successfully applied to a wide range of properties in non-parametric stochastic systems, including reachability~\cite{ChakarovS13,ChatterjeeFNH18,ChatterjeeFG16,McIverMKK18,AbateGR20,ChatterjeeGMZ22,LechnerZCH22,ChatterjeeGNZZ23,MajumdarS25,KuraUT26}, safety~\cite{PrajnaJP07,SteinhardtT12,ChatterjeeNZ17,TakisakaOUH21,ZhiWLOZ24,WangS0CG21,AbateEGPR23,WangYFLO24}, reach-avoidance~\cite{XueLZF21,ZikelicLHC23,ZikelicLVCH23,Xue24}, and $\omega$-regular specifications \cite{AbateGR24,HenzingerMSZ25,AbateGR25,KuraU26,AbateGIR26}. However, they have not been considered in the parametric setting.

Our first contribution is the {\em parameter flattening transformation} of pMDPs that embeds parameter valuations into the state space, allowing us to transform pMDPs into semantically equivalent non-parametric MDPs. The transformation is applicable to pMDPs with general measurable state and action spaces, that may be both finite and infinite sets. This transformation, that has previously not been explored in the literature, opens an opportunity to bring the techniques for formal verification of non-parametric MDPs to the setting of  pMDPs.

Our second contribution is the introduction of {\em parametric supermartingale certificates}. We develop these certificates by instantiating existing non-parametric supermartingale certificates in the parameter flattened MDP. This yields a new toolkit for reasoning about pMDPs that fundamentally differs from the existing techniques for finite-state pMDP model checking, and provides a general framework for generalizing supermartingale certificates to pMDPs with general measurable state and action spaces.

Our third contribution consists of fully automated algorithms for verification
and approximate synthesis in \textit{polynomial} pMDPs, i.e.~pMDPs whose state-, action-
and parameter-spaces and the stochastic kernel can be expressed in polynomial
real arithmetic. Using template-based synthesis and SMT solving, we
automatically construct parametric supermartingale certificates for
verification. For synthesis, we employ an abstraction–refinement procedure that
partitions the parameter space into regions of satisfaction and violation.
Alongside, as a witness of the satisfaction of the specification, both our verification and synthesis algorithms produce a
\textit{parametric strategy} in a closed-form expression that explicitly captures the strategy as a symbolic polynomial function of a parameter valuation,
which can be instantiated to concrete strategies by substituting in the parameter valuation.
To the best of our knowledge, this is the first algorithm that produces such parametric strategies in the literature on pMDPs.

Finally, we implement a prototype and evaluate it on infinite-state pMDP benchmarks modeling continuous random walk scenarios.




\smallskip\noindent{\bf Related work.} 
The verification of stochastic systems has a long and well-established history, encompassing the analysis of Markov chains as well as the synthesis and verification of policies for Markov decision processes (MDPs) and partially observable MDPs (POMDPs). More recently, significant attention has been devoted to stochastic systems with uncertain transition probabilities in \textit{finite}-state domains. This line of work broadly falls into two categories: models with interval-valued probabilities~\cite{chatterjee2008model} and models with parametric probability distributions~\cite{JansenJK22}. Parametric models are strictly more expressive, as appropriate choices of parameters can also represent probability intervals.

The study of pMDPs over infinite or continuous state spaces remains limited. To the best of our knowledge, the only known approach~\cite{peruffo2021formal} is restricted to abstraction-based techniques that support bounded-horizon properties, thereby limiting their applicability to long-run or asymptotic behaviors. Independently, supermartingale-based certificates have played a central role in the verification and synthesis of stochastic systems. These techniques apply to general (including infinite and continuous) state spaces and support a broad class of specifications, ranging from safety~\cite{PrajnaJP07,SteinhardtT12,ChatterjeeNZ17,TakisakaOUH21,ZhiWLOZ24,WangS0CG21,AbateEGPR23,WangYFLO24} and reachability~\cite{ChakarovS13,ChatterjeeFNH18,ChatterjeeFG16,McIverMKK18,AbateGR20,ChatterjeeGMZ22,LechnerZCH22,ChatterjeeGNZZ23,MajumdarS25,KuraUT26} to more expressive $\omega$-regular properties~\cite{HenzingerMSZ25,AbateGR24,AbateGR25,KuraU26,AbateGIR26}. Our work bridges two lines of research by combining supermartingale certificates with parametric stochastic models over infinite state spaces, thereby enabling verification of rich behavioral properties in the presence of parametric uncertainty.
\section{Preliminaries}\label{sec:prelims}

We assume that the reader is familiar with basic notions of measure and probability theory such as measurable space, probability measure, or expected value, see~\cite{williams1991probability} for formal definitions of these notions. For a measurable space $(X, \mathcal{F}_X)$ consisting of a set $X$ and a $\sigma$-algebra $\mathcal{F}_X$ over $X$, we use $\Delta(X,\mathcal{F}_X)$ to denote the set of all probability distributions over $(X,\mathcal{F}_X)$. For a probability distribution $\delta \in \Delta(X,\mathcal{F}_X)$, we use $\support(\delta)$ to denote its support. 
For a set $X$, we use 
$X^*$ and $X^\omega$ to denote the sets of all finite and infinite words over $X$.


\smallskip\noindent{\bf MDPs.} We now formally define Markov decision processes (MDPs) over general measurable state and action spaces. A (measurable) MDP is a tuple $M = (\States, \StatesSA, \Actions, \ActionsSA, \kernel, \Init)$, where
\begin{compactitem}
	\item $(\States, \StatesSA)$ is the {\em state space}, with $\States$ being the set of {\em states} and $\StatesSA$ being a $\sigma$-algebra over the set $\States$, with $\Init \subseteq \States$ being the set of {\em initial states},
	\item $(\Actions, \ActionsSA)$ is the {\em action space}, with $\Actions$ being the set of {\em actions} and $\ActionsSA$ being a $\sigma$-algebra over the set $\Actions$, and
	\item $\kernel: \States \times \Actions \times \StatesSA \rightarrow [0,1]$ is the {\em stochastic kernel} (or {\em probabilistic transition function}), where we require that (1) $\kernel(s, a, \cdot): \StatesSA \rightarrow [0,1]$ is a probability measure for each state-action pair $(s,a) \in \States \times \Actions$, and (2) $\kernel(\cdot, \cdot, B): \States \times \Actions \rightarrow [0,1]$ is a measurable function for each event $B \in \StatesSA$.
\end{compactitem}

For each state-action pair $(s,a) \in \States \times \Actions$, we sometimes use the notation $\kernel(\cdot \mid s,a) = \kernel(s,a,\cdot): \StatesSA \rightarrow [0,1]$ to denote the probability distribution over the state space induced by the stochastic kernel at the state-action pair $(s,a)$. In the special case when the action set is a singleton, i.e.~$|\Actions| = 1$ and $\ActionsSA = \{\emptyset,\Actions\}$, we say that $\MDP$ is a (measurable) {\em Markov chain}.

A {\em path} is an infinite sequence of state-action pairs $s_0,a_0,s_1,a_1, \dots \in (\States \times \Actions)^\omega$ such that $s_0 \in \Init$ and $s_{i+1} \in \support(P(\cdot \mid s_i,a_i))$ for each $i \in \mathbb{N}_0$. A {\em finite path} $s_0,a_0,s_1,a_1,\dots,s_i$ is a finite prefix of a path that ends in a state. We use $\Path^\MDP$ and $\FPath^\MDP$ to denote the sets of all paths and finite paths in the MDP $\MDP$.

\smallskip\noindent{\bf Strategies and semantics of MDPs.} Similarly to finite-state MDPs, the semantics of general measurable MDPs are formalized via the notion of strategies~\cite{Puterman94}.  A {\em strategy} (also known as a {\em policy} or a {\em controller}) in an MDP is a measurable map $\sigma: \States \rightarrow \Actions$ which to each state in the MDP assigns an action. We use $\StrategySpace^\MDP$ to denote the set of all strategies in the MDP $\MDP$. Note that in this work we restrict our attention to deterministic memoryless strategies, i.e.~strategies whose action assignment only depends on the current state and does not use information about the past, and which do not randomize between actions. We discuss the implications of this restriction in Remark~\ref{rmk:memoryless}.

The semantics of an MDP are defined as follows. For each initial state $
s \in \Init$, an MDP $\MDP$ and a strategy $\sigma$ together define a probability space over the set of all MDP paths. A random path in this probability space is generated as follows. The initial state of the path is $s_0 = s$. For each $i \in \mathbb{N}_0$, let the $i$-th state of the path be $s_i$. The $i$-th action $a_i$ of the path is then specified by the strategy via $a_i = \sigma(s_i)$, and the $(i+1)$-st state $s_{i+1}$ is sampled from the probability distribution $P(\cdot \mid s_i,a_i)$ specified by the stochastic kernel. The process is executed ad infinitum and gives rise to a random infinite path $s_0,a_0,s_1,a_1,\dots \in \Path^\MDP$ in the MDP.

This process is formalized via the cylinder construction~\cite[Theorem~2.7.2]{ash2000probability} and it gives rise to a probability space $(\Path^\MDP,\mathcal{F}_{\Path^\MDP},\mathbb{P}_s^{\MDP,\sigma})$ over the set of all MDP paths, where $\mathcal{F}_{\Path^\MDP}$ is a $\sigma$-algebra over the paths and $\mathbb{P}_s^{\MDP,\sigma}$ is a probability measure. We use $\mathbb{E}_s^{\MDP,\sigma}$ to denote expectation operator in this probability space.

\smallskip\noindent{\bf Parametric MDPs (pMDP).} We now define pMDPs,
which are the central object of study in this work. Intuitively, a pMDP is an MDP whose stochastic kernel is parametrized by finitely many
real-valued parameters, and each parameter valuation gives rise to one concrete
MDP with the stochastic kernel defined by the parameter valuation. Formally, a
{\em parametric MDP (pMDP)} is a tuple $M = (\States, \StatesSA, \Actions, \ActionsSA,
\kernel, \Init, \ParamSet)$, where each component is defined analogously as
above in the case of MDPs, with two exceptions, namely
\begin{compactitem}
	\item $U \subseteq \mathbb{R}^n$ is a Borel-measurable set of {\em parameter valuations}, where $n \in \mathbb{N}_0$ is the number of real-valued {\em parameters}, and
	\item $\kernel: \ParamSet \times \States \times \Actions \times \StatesSA \rightarrow [0,1]$ is the {\em parametric stochastic kernel} (or {\em parametric probabilistic transition function}), where we require that (1) for each parameter valuation $u \in \ParamSet$, $\kernel(u, \cdot, \cdot, \cdot): \States \times \Actions \times \StatesSA \rightarrow [0,1]$ is a stochastic kernel as defined above, and (2) $\kernel(\cdot, s, a, B): U \rightarrow [0,1]$ is a measurable function, where $\ParamSet$ is equipped with the Borel $\sigma$-algebra.
\end{compactitem}
Intuitively, each parameter valuation $u \in \ParamSet$ produces a stochastic kernel $\kernel[u] = \kernel(u,\cdot,\cdot,\cdot): \States \times \Actions \times \StatesSA \rightarrow [0,1]$, which in turn gives rise to a (non-parametric) MDP $\MDP[u] = (\States, \StatesSA, \Actions, \ActionsSA, \kernel[u], \Init)$.
Moreover, (non-parametric) MDPs are a special case of pMDPs, where the stochastic kernel $P=P[u]$ for all $u\in U$.

\begin{example}[Running example]\label{ex:running}
	We define a pMDP $M$ which will serve as our running example throughout this work. The state space $(\States,\StatesSA)$ of this pMDP is the interval of real values $\States = [-2,150]$ equipped with the Borel $\sigma$-algebra $\StatesSA = \mathcal{B}([-2,150])$, and the action space $(\Actions,\ActionsSA)$ is the interval of real values $\Actions = [0,0.6]$ equipped with the Borel $\sigma$-algebra $\ActionsSA = \mathcal{B}([0,0.6])$. The pMDP contains $2$ real-valued parameters and the set of parameter valuations is $\ParamSet = [-1,1]^2 = \{(u_1,u_2) \mid -1 \leq u_1, u_2 \leq 1\}$. Finally, the parametric stochastic kernel is defined by the equation
	\begin{equation}\label{eq:runningexample}
		s' \coloneqq 
		\begin{cases}
			s + a + u_1 + u_2 \cdot \Uniform[-0.75,0.25],	&	\text{if }0 \leq s \leq 100,\\
			s,							&	\text{otherwise},
		\end{cases}
	\end{equation}
\end{example}
where $\Uniform[-0.75,0.25]$ denotes the uniform distribution over the real interval $[-0.75,0.25]$. In other words, for each parameter valuation $(u_1,u_2) \in \ParamSet$, state $s \in \States$ and action $a \in \Actions$, the probability distribution $\kernel(\cdot \mid u, s, a)$ induced by the parametric stochastic kernel is defined as in eq.~\eqref{eq:runningexample}. Finally, the set of initial states is the interval of real values $\Init = [2,3]$.

\smallskip\noindent{\bf Parametric strategies and semantics of pMDPs.} The semantics of pMDPs are formalized via the notion of parametric strategies. Intuitively, a parametric strategy in a pMDP is an object which, for each parameter valuation, gives rise to one concrete map that to each state assigns an action. Formally, a {\em parametric strategy} in a pMDP $\MDP$ is a measurable map $\sigma: \ParamSet \times \States \rightarrow \Actions$  such that, for each parameter valuation $u \in \ParamSet$, the map $\sigma[u] = \sigma(u,\cdot): \States \rightarrow \Actions$ is a strategy in the MDP $\MDP[u]$. We use $\StrategySpace^\MDP$ to denote the set of all parametric strategies in $\MDP$. 

For each parameter valuation $u \in \ParamSet$, the MDP $\MDP[u]$, initial state $s \in \Init$ and parametric strategy $\sigma[u]$ together gives rise to a probability space $(\Path^{\MDP[u]},\mathcal{F}_{\Path^{\MDP[u]}},\mathbb{P}_s^{\MDP[u],\sigma[u]})$ defined analogously as above. We use $\mathbb{E}_s^{\MDP[u],\sigma[u]}$ to denote the expectation operator in this probability space.

	
\smallskip\noindent{\bf Specifications.} A {\em specification} $\phi$ in a pMDP $\MDP$ is a set of paths $\phi$ in $\Path^\MDP$. For the semantics of our problem to be well-defined, we require that this set of paths belongs to the $\sigma$-algebra induced by the pMDP semantics, i.e.~that $\phi \in \mathcal{F}_{\Path^{\MDP[u]}}$ for every $u \in U$. We say that a path $\rho \in \Path^\MDP$ {\em satisfies} the specification $\phi$, if $\rho \in \phi$. For each parameter valuation $u \in U$, initial state $s \in \Init$, and parametric strategy $\sigma$, we write $\mathbb{P}_s^{\MDP[u],\sigma[u]}[\phi]$ to denote the probability that a path in $\MDP$ sampled from the underlying probability space satisfies $\phi$. We use $\neg \phi = \Path^\MDP \backslash \phi$ to denote the specification which is satisfied by all paths that {\em do not} satisfy $\phi$ and vice versa.

Our examples will mostly focus on reachability and safety specifications. Let $X \in \StatesSA$ be a measurable set of MDP states. We define the {\em reachability specification} with respect to the set $X$ via $\Reach(X) = \{s_0,a_0,s_1,a_1,\dots \in \Path^\MDP \mid \exists i \in \mathbb{N}_0.\, s_i \in X\}$,
and the {\em safety specification} with respect to the set $X$ via $\Safe(X) = \{s_0,a_0,s_1,a_1,\dots \in \Path^\MDP \mid \forall i \in \mathbb{N}_0.\, s_i \in X\}$. 

\section{Problem Statement}\label{sec:problem}

We now define the parameter verification and synthesis problems that we consider in this work. In what follows, suppose that $\MDP = (\States, \StatesSA, \Actions, \ActionsSA, \kernel, \Init, \ParamSet)$ is a pMDP with $\ParamSet \subseteq \mathbb{R}^n$ being the set of parameter valuations, and let $\phi$ be a specification in $\MDP$. Following the terminology in~\cite{JungesAHJKQV24}, we say that a {\em parameter region} is a Borel-measurable subset $R \subseteq U$ of parameter valuations. We use $\MDP^R = (\States, \StatesSA, \Actions, \ActionsSA, \kernel, \Init, R)$ to denote the pMDP identical to $\MDP$ but whose set of parameter valuations is restricted to $R$.

\smallskip\noindent{\bf Specifications in pMDPs.} Section~\ref{sec:prelims}
formalizes specifications and the notion of paths satisfying specifications.
However, a pMDP gives rise to a probability space over the set of all paths for
each parameter valuation. Hence, when reasoning about specification satisfaction
in pMDPs, we are interested in deriving lower (and, dually, upper) bounds on the
probability of a random path in this probability space satisfying the
specification. Formally, we distinguish between {\em angelic} and {\em demonic}
satisfaction of the specification, depending on whether the specification needs
to be satisfied {\em for some} or {\em for all} parametric strategies in the
pMDP. Our definitions generalize the definitions of~\cite{JungesAHJKQV24} from
the setting of finite-state pMDPs to the setting of general measurable pMDPs:
\begin{compactenum}
	\item {\em Angelic satisfaction.} We say that the pMDP $\MDP$ {\em angelically satisfies} specification $\phi$ with probability at least $p \in [0,1]$, denoted $\MDP \models_a \phi, p$, if there exists a parametric strategy $\sigma: \ParamSet \times \States \rightarrow \Actions$ that leads to specification satisfaction for all parameter valuations, i.e.
	\begin{equation}\label{eq:angelic}
		\exists \sigma \in \StrategySpace^\MDP.\,\,\, \forall u \in U.\,\,\, \forall s \in \Init.\,\,\, \mathbb{P}_s^{\MDP[u],\sigma[u]}[\phi] \geq p.
	\end{equation}
	\item {\em Demonic satisfaction.} We say that the pMDP $\MDP$ {\em demonically satisfies} specification $\phi$ with probability at least $p \in [0,1]$, denoted $\MDP \models_d \phi, p$, if all parametric strategies $\sigma: \ParamSet \times \States \rightarrow \Actions$ lead to specification satisfaction for all parameter valuations, i.e.
	\begin{equation}\label{eq:demonic}
		\forall \sigma \in \StrategySpace^\MDP.\,\,\, \forall u \in U.\,\,\, \forall s \in \Init.\,\,\, \mathbb{P}_s^{\MDP[u],\sigma[u]}[\phi] \geq p.
	\end{equation}
\end{compactenum}
In both cases, quantification is over {\em parametric strategies} of type $\sigma: U \times S \rightarrow A$, which for each parameter valuation induce one regular, non-parametric strategy.

\smallskip\noindent{\bf Problems.} We consider the following problems, as discussed in Section~\ref{sec:intro}:
\begin{compactenum}
	\item {\bf Verification problem.} Given a pMDP $\MDP$ with a set of parameter valuations $U$, a specifciation $\phi$, a probability threshold $p \in [0,1]$, and a type of specification satisfaction $\diamondsuit \in \{a,d\}$, prove that $\MDP \models_\diamondsuit \phi, p$.
	\item {\bf Synthesis problem.} Given a pMDP $\MDP$ with a set of parameter valuations $U$, a specifciation $\phi$, a probability threshold $p \in [0,1]$, and a type of specification satisfaction $\diamondsuit \in \{a,d\}$, partition the parameter set $U$ into two regions $R$ and $U \backslash R$ such that $\MDP^R \models_\diamondsuit \phi, p$ and $\MDP^{U \backslash R} \models_{\clubsuit} \neg\phi, 1-p$, where $\clubsuit$ is the opposite symbol of $\diamondsuit$ in $\{a,d\}$. 
	\item {\bf Approximate synthesis problem.} Given a pMDP $\MDP$ with a set of parameter valuations $U$, a specifciation $\phi$, a probability threshold $p \in [0,1]$, a type of specification satisfaction $\diamondsuit \in \{a,d\}$, and an approximation parameter $c \in [0,1]$, partition the parameter set $U$ into two disjoint regions $R$ and $R'$ such that $\MDP^R \models_\diamondsuit \phi, p$ and $\MDP^{R'} \models_{\clubsuit} \neg\phi, 1-p$, where $\clubsuit$ is the opposite symbol of $\diamondsuit$ in $\{a,d\}$, and such that $(\mu(R) + \mu(R')) / \mu(U) \geq 1 - c$ where $\mu$ is the Lebesgue measure over the parameter space.
\end{compactenum}

\begin{example}\label{ex:problem}
	Consider the pMDP $\MDP$ with parameter valuations $U = [-1,1]^2$ as in Example~\ref{ex:running}. Let $\phi = \Reach([90,150])$, i.e.~a reachability specification for the target set $T = [90,150]$, and let $p = 0.9$. For the sake of the example, suppose that we are interested in angelic satisfaction, i.e.~$\diamondsuit = a$.
	
	Given a parameter region $R \subseteq U$, the verification problem is concerned with proving that there exists a parametric strategy s.t., for each $(u_1,u_2) \in R$ and initial state $s \in [2,3]$, $\Reach([90,150])$ is satisfied with probability at least $0.9$.
	
	On the other hand, the (approximate) synthesis problem is concerned with partitioning $U$ into two parameter regions $R$ and $R'$ such that (1)~there exists a parametric strategy s.t., for each $(u_1,u_2) \in R$ and initial state $s \in [2,3]$, $\Reach([90,150])$ is satisfied with probability at least $0.9$, while (2)~for all parametric strategies, for each $(u_1,u_2) \in R$ and initial state $s \in [2,3]$, $\neg\Reach([90,150]) = \Safe([-2,90))$ is satisfied with probability at least $1 - 0.9 = 0.1$. Thus, for parameter valuations in $R$ we require satisfaction of an angelic reachability specification with probability at least $0.9$, while for parameter valuations in $R'$ we require satisfaction of a demonic safety specification with probability at least $0.1$.
\end{example}

\begin{remark}[Consistency with finite-state pMDP definitions]
	It may seem that our definitions of angelic and demonic satisfaction in
	eq.~\eqref{eq:angelic} and~\eqref{eq:demonic} syntactically deviate from
	those presented in~\cite{JungesAHJKQV24} for finite-state pMDPs, which
	require: 
	\begin{alignat}{3}
		&\textbf{[angelic:]}\quad &&\forall u \in U.\,\,\, \exists
		\strategy[u] \in \StrategySpace^{\MDP[u]}.\,\,\, \forall s \in
		\Init.\,\,\, \mathbb{P}_s^{\MDP[u],\sigma[u]}[\phi] \geq p,\label{eq:aux1}\\
		&\textbf{[demonic:]}\quad &&\forall u \in U.\,\,\, \forall \strategy[u] \in \StrategySpace^{\MDP[u]}.\,\,\, \forall s \in \Init.\,\,\, \mathbb{P}_s^{\MDP[u],\sigma[u]}[\phi] \geq p.\label{eq:aux2}
	\end{alignat}
	%
	In reality, our definitions in eq.~\eqref{eq:angelic} and~\eqref{eq:demonic}
	{\em generalize} these definitions to the setting of general measurable
	pMDPs, such that they are mathematically well-defined and satisfy all
	measurability requirements, while remaining {\em equivalent} to the
	definitions of~\cite{JungesAHJKQV24} for \textit{finite-state} pMDPs. This
	is because our definitions in eq.~\eqref{eq:angelic} and~\eqref{eq:demonic}
	consider existential and universal quantifications over {\em parametric
	strategies}, but we still require that for each $u \in U$ there exists
	$\sigma[u]$ for angelic satisfaction, and that for each $u \in U$ demonic
	satisfaction holds for all $\sigma[u] \in U$. On the other hand, the
	existential and universal quantification over parametric strategies is
	necessary to ensure that the resulting strategies are {\em measurable} (note
	that this is a condition in our definition of parametric strategies), while
	the conditions in eq.~\eqref{eq:aux1} and~\eqref{eq:aux2} do not guarantee
	that the induced parametric strategies are measurable. 
	Eq.~\eqref{eq:aux1} guarantees the existence of a strategy $\sigma[u]$ in
	$\MDP[u]$ for each parameter valuation $u \in \ParamSet$, the parametric
	strategy $\sigma$ in the pMDP defined via $\sigma(u,s) := \sigma[u](s)$ for
	every $u \in \ParamSet$ and $s \in \States$ may not give rise to a
	measurable function.
	However, in the setting of finite-state pMDPs, the two definitions coincide since we consider the discrete topology over the finite state and action spaces, so each parametric strategy $\sigma$ in $\MDP$ gives rise to one strategy $\sigma[u]$ in $\MDP[u]$ for every parameter valuation $u \in \ParamSet$ and vice versa.
\end{remark}

\begin{remark}[Restriction to deterministic memoryless strategies]\label{rmk:memoryless}
	As stated in Section~\ref{sec:prelims}, we restrict our attention to deterministic memoryless strategies. The question of whether deterministic memoryless strategies are sufficient to satisfy a specification, both in MDPs and pMDPs, depends on the specification of interest. For instance, it is known that whenever there exists a strategy in an MDP satisfying a reachability specification, then there exists a memoryless deterministic strategy satisfying the specification. However, this statement does not hold for general omega-regular specifications in MDPs. In Proposition~\ref{prop:eq}, we show that our parameter flattening transformation is semantics preserving, which implies that deterministic memoryless parametric strategies are sufficient in pMDPs whenever deterministic memoryless strategies are sufficient in non-parametric MDPs.
\end{remark}
\section{Parameter Flattening}\label{sec:transformation}

The heart of our approach is parameter flattening, which is a novel
transformation of pMDPs into their semantically equivalent non-parametric counterparts.
This transformation allows us to take \textit{any} existing verification or
synthesis tool for non-parametric MDPs, and use it for solving the analogous problem for pMDPs.
In our work, we will use supermartingale certificates built for non-parametric
MDPs to design parametric supermartingale certificates for pMDPs; this will be
shown in Section~\ref{sec:theory} and \ref{sec:algo}.

Given a pMDP $\MDP$, the \textit{parameter flattening} transformation gives rise to a non-parametric MDP $\tilde{\MDP}$, which we call the \textit{parameter flattened MDP} of $\MDP$. Intuitively, the parameter flattened MDP $\tilde{\MDP}$ is obtained by incorporating parameter valuations into the state space of the original pMDP $\MDP$. Hence, each state in the parameter flattened MDP is a pair of a state and a parameter valuation in the original pMDP. On the other hand, the stochastic kernel of the parameter flattened MDP for each state-parameter valuation pair and for each action samples the successor state according to the stochastic kernel of the original pMDP $\MDP$, {\em while keeping the parameter valuation unchanged}. The following definition formalizes this informal construction.

\begin{definition}[Parameter flattening]\label{def:parameterflattening}
	Let $\MDP = (\States, \StatesSA, \Actions, \ActionsSA, \kernel, \Init, \ParamSet)$ be a pMDP. A {\em parameter flattened MDP} of $\MDP$ is the non-parametric MDP $\tilde{\MDP} = (\tilde{\States}, \tilde{\StatesSA}, \tilde{\Actions}, \tilde{\ActionsSA}, \tilde{\kernel}, \tilde{\Init})$ whose each component is defined as follows:
	\begin{compactitem}
		\item $\tilde{\States} = \States \times \ParamSet$ and $\tilde{\StatesSA} = \StatesSA \times \mathcal{B}(\ParamSet)$ where $\mathcal{B}(\ParamSet)$ is the Borel $\sigma$-algebra over $\ParamSet$, i.e.~the state space of the parameter flattened MDP is the product of the state space and the parameter space of the original pMDP.
		\item $\tilde{\Actions} = \Actions$ and $\tilde{\ActionsSA} = \ActionsSA$, i.e.~the action space of the parameter flattened MDP coincides with the action space of the original pMDP.
		\item $\tilde{\kernel}: \tilde{\States} \times \tilde{\Actions} \times \tilde{\StatesSA} \rightarrow [0,1]$ is defined as follows: For all $(s,u) \in \tilde{\States} = \States \times \ParamSet$, $a \in \tilde{\Actions} = \Actions$ and $(B,X) \in \tilde{\StatesSA} = \StatesSA \times \mathcal{B}(\ParamSet)$, we set $\tilde{\kernel}((s,u), a, (B,X)) = \kernel(u,s,a,B)$ if $u \in X$ and $\tilde{\kernel}((s,u), a, (B,X)) = 0$ if $u \not\in X$. That is, the stochastic kernel of the parameter flattened MDP induces the same probability distribution over the successor state as the original stochastic kernel {\em while keeping the parameter value $u$ unchanged} (so $u'=u$).
		\item $\tilde{\Init} = \Init \times \ParamSet$, i.e.~the set of the initial states in the parameter flattened MDP is the product of the set of initial states and the set of parameter valuations of the original pMDP. 
	\end{compactitem}
\end{definition}

\begin{example}
	Consider the pMDP $\MDP$ in Example~\ref{ex:running}. The parameter flattened MDP $\tilde{\MDP}$ of $\MDP$ is defined as follows. The state space is $\tilde{\States} = [-2,150] \times [-1,1]^2$, i.e.~the product of the state space and parameter valuations in $\MDP$, with the Borel $\sigma$-algebra $\tilde{\StatesSA} = \mathcal{B}([-2,150]) \times \mathcal{B}([-1,1]^2) = \mathcal{B}([-2,150] \times [-1,1]^2)$. The action space is the same as in $\MDP$, i.e. $\tilde{\Actions} = [0,0.6]$ and $\tilde{\ActionsSA} = \mathcal{B}([0,0.6])$. For each state $(s,u_1,u_2) \in \tilde{\States} = [-2,150] \times [-1,1]^2$ and for each action $a \in \tilde{\Actions} = [0,0.6]$, the probability distribution $\tilde{\kernel}(\cdot \mid (s, u_1,u_2), a)$ induced by the stochastic kernel is again defined by eq.~\eqref{eq:runningexample}. Finally, the initial set is $\tilde{\Init} = [2,3] \times [-1,1]^2$.
\end{example}

Every path $s_0,a_0,s_1,a_1, \dots$ and parameter valuation $u \in \ParamSet$ in~$\MDP$ together induce a path $(s_0,u),a_0,(s_1,u),a_1, \dots$ in the parameter flattened MDP $\tilde{\MDP}$. Hence, a specification $\phi$ in $\MDP$ induces a {\em parameter flattened specification} $\tilde{\phi}$ in the parameter flattened MDP $\tilde{\MDP}$ defined via $\tilde{\phi} = \{(s_0,u),a_0,(s_1,u),a_1, \dots \in \Path^{\tilde{\MDP}} \mid s_0,a_0,s_1,a_1, \dots \in \phi \land u \in \ParamSet\}$.
On the other hand, every parametric strategy $\sigma: U \times \States \rightarrow \Actions$ in the pMDP $\MDP$ induces a {\em parameter flattened strategy} $\tilde{\sigma}: \tilde{\States} \rightarrow \tilde{\Actions}$ via $\tilde{\sigma}(s,u) = \sigma[u](s)$ for all $(s,u) \in \tilde{\States} = \States \times \ParamSet$.

The following proposition establishes the semantic equivalence between pMDPs and their parameter flattened MDPs. The proof is in Appendix~\ref{app:propositionproof}.

\begin{proposition}[Parameter flattening is semantic preserving]\label{prop:eq}
	Let $\MDP = (\States, \StatesSA, \Actions, \ActionsSA, \kernel, \Init, \ParamSet)$ be a pMDP, $\sigma: \ParamSet \times \States \rightarrow \Actions$ be a parametric strategy in $\MDP$, and $\phi$ be a specification in $\MDP$. Let $\tilde{\MDP} = (\tilde{\States}, \tilde{\StatesSA}, \tilde{\Actions}, \tilde{\ActionsSA}, \tilde{\kernel}, \tilde{\Init})$, $\tilde{\sigma}: \tilde{\States} \rightarrow \tilde{\Actions}$ and $\tilde{\phi}$ be the parameter flattened MDP of $\MDP$, the parameter flattened strategy of $\sigma$ and the parameter flattened specification of $\phi$, respectively. Then, for each initial state $s \in \Init$ and each parameter valuation $u \in U$ in the pMDP $\MDP$, 
	\[ \mathbb{P}_s^{\MDP[u],\sigma[u]}[\phi] = \mathbb{P}_{(s,u)}^{\tilde{\MDP},\tilde{\strategy}}[\tilde{\phi}]. \]
\end{proposition}

The following theorem is then a corollary of the above proposition and of our definitions of angelic and demonic satisfaction of specifications in parametric and non-parametric MDPs. We defer the proof to Appendix~\ref{app:transformation}.

\begin{theorem}[Soundness and completeness]\label{thm:transformation}
	Let $\MDP$ be a pMDP, $\tilde{\MDP}$ be the parameter flattened
	MDP of $\MDP$, $\phi$ be a specification in $\phi$, $p \in [0,1]$ be a
	probability threshold, and $\diamondsuit \in \{a,d\}$ denote either angelic
	or demonic satisfaction. Then, we have $\MDP \models_\diamondsuit \phi,p$ if
	and only if $\tilde{\MDP} \models_\diamondsuit \tilde{\phi},p$.
\end{theorem}

Hence, in order to prove specification satisfaction in a pMDP, it suffices to prove satisfaction in the parameter flattened MDP. This motivates our construction of parametric supermartingale certificates as instantiations of supermartingale certificates in the state space of parameter flattened MDPs.


\section{Parametric Supermartingale Certificates}\label{sec:theory}

We now present our novel parametric supermartingale certificates for proving
satisfaction of specifications in pMDPs. Our parametric
supermartingale certificates generalize the existing supermartingale
certificates for non-parametric MDPs by incorporating the parameters within the
definitions of the certificates.
This is done in such a way that the new parametric certificates can be computed
by {\em instantiating} non-parametric supermartingale certificates in the state
space of the parameter flattened MDP. While this leads to an elegant framework
for generalizing most existing supermartingale certificates to the setting of
pMDPs, it is technically non-trivial and requires care in
appropriately handling the quantification of parameter valuations and
distinguishing between angelic and demonic satisfaction of specifications in
pMDPs.
To the best of our knowledge, this is the first work that uses supermartingale
certificates to reason about pMDPs.

This section is organized as follows. We illustrate the definition
of parametric supermartingale certificates via instantiation in the state space
of the parameter flattened MDP on two classes of commonly used supermartingale
certificates which will also be used in our experimental evaluation.
Specifically, we present pMDP generalizations of stochastic barrier
functions for probability $p \in [0,1)$ safety~\cite{PrajnaJP07}
(Section~\ref{sec:parametricsbf}) and of reach-avoid supermartingales for
probability $p \in [0,1)$ reachability~\cite{ZikelicLHC23}
(Section~\ref{sec:parametricrasm}). In Appendix~\ref{app:parametricrsm}, we show
a pMDP generalization of another established supermartingale
certificate, namely ranking supermartingales (RSMs) for probability~1 reachability~\cite{ChakarovS13}. 

While we use these established supermartingale certificates to illustrate our
construction of parametric supermartingale certificates, we stress that our
construction applies to other supermartingale certificates as well, allowing us
to generalize a broad class of supermartingale certificates to the setting of
pMDPs. These include the recent certificates for general $\omega$-regular
specifications in non-parametric MDPs~\cite{AbateGR24,HenzingerMSZ25,AbateGR25}. 
Appendix~\ref{app:otherspecifications} formalizes the class of supermartingale
certificates that can be generalized for pMDPs in the same manner.

\smallskip\noindent{\bf Invariants.} Supermartingale certificates are functions that map MDP states to real values that are required to satisfy a set of conditions at every reachable MDP state. Since computing the set of reachable states in a general measurable MDP is computationally intractable, we define supermartingale certificates with respect to an over-approximation of the set of all reachable states which is called a (supporting) invariant. Formally, a state in an MDP is said to be {\em reachable} if it is contained in some MDP path. An {\em invariant} in an MDP is a measurable set $I$ of MDP states that contains all reachable states in the MDP. We note that, while in this section we assume the existence of a supporting invariant, our verification and synthesis algorithms in Section~\ref{sec:algo} will compute the invariant together with the parametric supermartingale certificate.

\subsection{Parametric Supermartingales for Safety
Specifications}\label{sec:parametricsbf}

For proving probability $p \in [0,1)$ (a.k.a.~quantitative) safety
specifications, {\em stochastic barrier functions (SBFs)} are an established
supermartingale certificate~\cite{PrajnaJP07}.
We present their generalization to the setting of pMDPs by instantiating them in the state space of parameter flattened MDPs.

We first recall (non-parametric) SBFs. Intuitively, given a Markov chain $\MDP$, a set of safe states $X$ that we want to prove the system always remains in with probability at least $p$, and an invariant $I$ that over-approximates the set of reachable states, an SBF is a measurable function $V: \States \rightarrow \mathbb{R}$ that to each state assigns a real value that is required to satisfy the following four conditions: {\bf (C1)} $V$ is nonnegative at all states in the invariant $I$, {\bf (C2)} $V(s) \leq 1$ at all initial states $s$ in the MDP, {\bf (C3)} $V(s) \geq \frac{1}{1-p}$ at all invariant states $s$ outside of the safe set~$X$, and {\bf (C4)} at all invariant states $s$ at which $V(s) \leq \frac{1}{1-p}$, the value of $V$ (non-strictly) decrease in expected value upon the one-step execution of the Markov chain. Thus, the value of the SBF is required to be nonnegative, initially be below $1$ and to decrease in expectation, while the safety specification can be violated only if the value of the SBF exceeds $\frac{1}{1-p}$ despite the expected decrease. It was shown in~\cite{PrajnaJP07} that these conditions together ensure that the safety specification is satisfied with probability at least $p$.


The following definition formalizes the above intuition and generalizes it to the setting of pMDPs by instantiating it in the state space of parameter flattened MDP. In doing so, however, we need to exercise additional care to properly quantify parameter valuations, and to achieve the extension from Markov chains to MDPs by differentiating between angelic and demonic satisfaction. Note that, since parametric SBFs are defined over the state space of parameter flattened MDP, they are defined as measurable functions of type $V: \States \times \ParamSet \rightarrow \mathbb{R}$ that map {\em states of the parameter flattened MDP} to real values.

\begin{definition}[Parametric SBFs]\label{def:psbf}
	Let $\MDP = (\States, \StatesSA, \Actions, \ActionsSA, \kernel, \Init, \ParamSet)$ be a pMDP, $I \in \StatesSA$ be an invariant, $X \in \StatesSA$ be a set of safe states, $p \in [0,1)$ be a probability threshold, and $\diamondsuit \in \{a,d\}$ denote either angelic or demonic satisfaction. A {\em parametric stochastic barrier function (parametric SBF)} for $X$ with respect to $I$ and $\diamondsuit$ is a measurable function $V: \States \times \ParamSet \rightarrow \mathbb{R}$, which is required to satisfy the following four conditions:
	\begin{compactenum}
		\item {\em Nonnegativity.} $V(s,u) \geq 0$ for all $s \in I$ and $u \in U$.
		\item {\em Initial.} $V(s,u) \leq 1$ for all $s \in \Init$ and $u \in U$.
		\item {\em Safety.} $V(s,u) \geq \frac{1}{1-p}$ for all $s \in I \backslash X$ and $u \in U$.
		\item {\em Expected decrease.} We distinguish between angelic and demonic satisfaction:
		\begin{compactenum}
			\item {\em Angelic satisfaction, i.e.~$\diamondsuit = a$.} There exists a parametric strategy $\sigma: \ParamSet \times \States \rightarrow \Actions$ in $\MDP$ such that, for every parameter valuation $u \in U$ and for every state $s \in I$ with $V(s,u) \leq 1/(1-p)$, $V$ satisfies expected decrease for the action $a = \sigma[u](s)$, i.e.
			\begin{equation*}
			\begin{split}
			\exists \strategy \in \StrategySpace^\MDP.\,\,\,  \forall u \in \ParamSet.\,\,\, \forall s \in I.\,\,\, &V(s,u) \leq \frac{1}{1-p} \\
			&\Rightarrow V(s,u)  \geq \mathbb{E}_{s' \sim \kernel(\cdot \mid u,s,\sigma[u](s))}[V(s',u)].
			\end{split}
			\end{equation*}
			\item {\em Demonic satisfaction, i.e.~$\diamondsuit = d$.} For every parametric strategy $\sigma: \ParamSet \times \States \rightarrow \Actions$ in $\MDP$ and for every parameter valuation $u \in U$ and state $s \in I$ with $V(s,u) \leq 1/(1-p)$, we have that $V$ satisfies expected decrease, i.e.
				\begin{equation*}
				\begin{split}
					\forall \strategy \in \StrategySpace^\MDP.\,\,\,  \forall u \in \ParamSet.\,\,\, \forall s \in I.\,\,\, &V(s,u) \leq \frac{1}{1-p} \\
					&\Rightarrow V(s,u)  \geq \mathbb{E}_{s' \sim \kernel(\cdot \mid u,s,\sigma[u](s))}[V(s',u)].
				\end{split}
			\end{equation*}
		\end{compactenum}
	\end{compactenum}
\end{definition}

The following theorem shows that parametric SBFs provide a sound proof rule for proving probability~$p \in [0,1)$ safety in pMDPs. The proof proceeds by showing that the parametric SBF induces a non-parametric SBF in the parameter flattened MDP and then using the soundness of SBFs for proving probability $p \in [0,1)$ safety~\cite{PrajnaJP07}, while taking additional care to correctly handle the cases of angelic and demonic satisfaction. We defer it to Appendix~\ref{app:psbf}. 

\begin{theorem}[Soundness of parametric SBFs]\label{thm:psbf}
	Let $\MDP$ be a pMDP, $I$ be an invariant, $X$ be a set of safe states, $p \in [0,1)$, and $\diamondsuit \in \{a,d\}$ denote either angelic or demonic satisfaction. Suppose that there exists a parametric SBF for $V$ with respect to $I$ and $\diamondsuit$. Then, we have $\MDP \models_\diamondsuit \Safe(X), p$.
\end{theorem}


\begin{example}
	Consider the pMDP $\MDP$ in Example~\ref{ex:running} with the demonic safety specification $\phi = \Safe([-2,90])$, $\diamondsuit = d$, and probability threshold $p = 0.1$ as in Example~\ref{ex:problem}. Let $R = [-0.25, -0.125] \times \{1\}$ be a parameter region for which we want to prove that $\MDP^R \models_d \Safe([-2,90]), 0.1$. 
	An example of a parametric SBF and a supporting invariant computed by our prototype tool, which can be used to prove this statement, are given by $V(s,u_1,u_2) = -2 \cdot 10^{-5} - u_1 + 0.011 \cdot s$ and $I = [-2,150]$, i.e~the invariant is equal to the whole set of states. One can verify by inspection that $V$ satisfies all the defining conditions in Definition~\ref{def:psbf}. Hence, by Theorem~\ref{thm:psbf}, it provides a proof that $\MDP^R \models_d \Safe([-2,90]), 0.1$.
\end{example}

\begin{remark}[Probability $1$ safety]
	Our parametric SBFs are defined and proved to be sound for $p \in [0,1)$
	safety. However, for $p=1$, parametric SBFs become ill defined due to the
	quotient $\frac{1}{1-p}$ in their defining conditions. This issue already
	arises in the non-parametric setting, and is not an artifact of pMDPs. The
	more conceptual issues regarding the unsoundness of SBFs when $p = 1$ are
	discussed in~\cite{SoCF23}. In practice, to achieve soundness for $p=1$, one
	can strengthen almost-sure safety to sure safety and use classical
	(non-probabilistic) barrier functions. Their instantiation in the setting of
	pMDPs is achieved similarly to parametric SBFs in Definition~\ref{def:psbf},
	by considering the parameter flattened MDP.
\end{remark}

\subsection{Parametric Supermartingales for Reachability
Specifications}\label{sec:parametricrasm}

We now present a parametric spermartingale certificate generalization of {\em reach-avoid supermartingales (RASMs)} for proving probability $p \in [0,1)$ (a.k.a.~quantitative) reachability and reach-avoidance specifications~\cite{ZikelicLHC23}. Here we consider reachability since our focus in the experiments will be on reachability specifications, however an extansion to reach-avoid specifications can be achieved by including an additional safety condition analogous to condiiton~$3$ in Definition~\ref{def:psbf}.

We first recall (non-parametric) RASMs. Intuitively, given a Markov chain $\MDP$, a set of target states $T$ that we want to prove the system reaches with probability at least $p$, and an invariant $I$, a RASM is a measurable function $V: \States \rightarrow \mathbb{R}$ that to each state assigns a real value that is required to satisfy the following four conditions: {\bf (C1)} $V$ is nonnegative at all states in the invariant $I$, {\bf (C2)} $V(s) \leq 1$ at all initial states $s$ in the MDP, and {\bf (C3)} at all non-target invariant states $s \in I \backslash T$ at which $V(s) \leq \frac{1}{1-p}$, the value of $V$ strictly decrease by at least $\epsilon > 0$ in expected value upon the one-step execution of the Markov chain. Thus, the value of the SBF is required to be nonnegative, initially be below $1$ and to strictly $\epsilon$-decrease in expectation until it either reaches the target set or a state with $V(s) \geq \frac{1}{1-p}$. It was shown in~\cite{ZikelicLHC23} that these conditions together ensure that the reachability specification is satisfied with probability at least $p$.

The following definition formalizes the above intuition and generalizes it to the setting of pMDPs by instantiating it in the state space of parameter flattened MDP. Since parametric RASMs are defined over the state space of parameter flattened MDP, they are defined as measurable functions of type $V: \States \times \ParamSet \rightarrow \mathbb{R}$ that map {\em states of the parameter flattened MDP} to real values.

\begin{definition}[Parametric RASMs]\label{def:prasm}
	Let $\MDP = (\States, \StatesSA, \Actions, \ActionsSA, \kernel, \Init, \ParamSet)$ be a pMDP, $I \in \StatesSA$ be an invariant, $T \in \StatesSA$ be a set of target states, $p \in [0,1)$ be a probability threshold, and $\diamondsuit \in \{a,d\}$ denote either angelic or demonic satisfaction. A {\em parametric reach-avoid supermartingale (parametric RASM)} for $T$ with respect to $I$ and $\diamondsuit$ is a measurable function $V: \States \times \ParamSet \rightarrow \mathbb{R}$, which is required to satisfy the following four conditions:
	\begin{compactenum}
		\item {\em Nonnegativity.} $V(s,u) \geq 0$ for all $s \in I$ and $u \in U$.
		\item {\em Initial.} $V(s,u) \leq 1$ for all $s \in \Init$ and $u \in U$.
		\item {\em Expected decrease.} We distinguish between angelic and demonic satisfaction:
		\begin{compactenum}
			\item {\em Angelic satisfaction, i.e.~$\diamondsuit = a$.} There exist $\epsilon > 0$ and a parametric strategy $\sigma: \ParamSet \times \States \rightarrow \Actions$ in $\MDP$ such that, for every parameter valuation $u \in U$ and for every non-target state $s \in I \backslash T$ with $V(s,u) \leq 1/(1-p)$, $V$ satisfies expected decrease for the action $a = \sigma[u](s)$, i.e.
			\begin{equation*}
				\begin{split}
					\exists \strategy \in \StrategySpace^\MDP.\,\,\,  \forall u \in \ParamSet.\,\,\, &\forall s \in I.\,\,\, s \in I \backslash T \land V(s,u) \leq \frac{1}{1-p} \\
					&\Rightarrow V(s,u)  \geq \mathbb{E}_{s' \sim \kernel(\cdot \mid u,s,\sigma[u](s))}[V(s',u)] + \epsilon.
				\end{split}
			\end{equation*}
			\item {\em Demonic satisfaction, i.e.~$\diamondsuit = d$.} There exists $\epsilon > 0$ such that for all parametric strategies $\sigma: \ParamSet \times \States \rightarrow \Actions$ in $\MDP$, for every parameter valuation $u \in U$ and for every non-target state $s \in I \backslash T$ with $V(s,u) \leq 1/(1-p)$, $V$ satisfies expected decrease for the action $a = \sigma[u](s)$, i.e.
			\begin{equation*}
				\begin{split}
					\forall \strategy \in \StrategySpace^\MDP.\,\,\,  \forall u \in \ParamSet.\,\,\, &\forall s \in I.\,\,\, s \in I \backslash T \land V(s,u) \leq \frac{1}{1-p} \\
					&\Rightarrow V(s,u)  \geq \mathbb{E}_{s' \sim \kernel(\cdot \mid u,s,\sigma[u](s))}[V(s',u)] + \epsilon.
				\end{split}
			\end{equation*}
		\end{compactenum}
	\end{compactenum}
\end{definition}

The following theorem shows that parametric RASMs provide a sound proof rule for probability~$p \in [0,1)$ reachability in pMDPs; the proof is in Appendix~\ref{app:prasm}. 

\begin{theorem}[Soundness of parametric RASMs]\label{thm:prasm}
	Let $\MDP$ be a pMDP, $I$ be an invariant, $T$ be a set of target states, $p \in [0,1)$ be a probability threshold, and $\diamondsuit \in \{a,d\}$ denote either angelic or demonic satisfaction. Suppose that there exists a parametric RASM for $T$ with respect to $I$ and $\diamondsuit$. Then,  $\MDP \models_\diamondsuit \Reach(T), p$.
\end{theorem}

\begin{example}
	Consider the pMDP $\MDP$ in Example~\ref{ex:running} with the angelic reachability specification $\phi = \Reach([90,150])$, $\diamondsuit = a$, and probability threshold $p = 0.9$ as in Example~\ref{ex:problem}. Let $R = [0.25, 0.5] \times \{1\}$ be a parameter region for which we want to prove that $\MDP^R \models_a \Reach([90,150]), 0.9$. 
	An example of a parametric RASM and a supporting invariant computed by our prototype tool, which can be used to prove this statement, are given by $V(s,u_1,u_2) = 0.5 \cdot u_1 - 3.5 \cdot 10^{-14} \cdot s$ and $I = \{s \in [-2,150] \mid -9.875 + 98.752 \cdot s \geq 0\}$. One can verify by inspection that $V$ satisfies all the defining conditions in Definition~\ref{def:prasm}. An existentially quantified parametric strategy, which is also computed by our tool, is defined via $\sigma(s,u_1,u_2) = 0.6$. Hence, by Theorem~\ref{thm:prasm}, it provides a proof that $\MDP^R \models_a \Reach([90,150]), 0.9$.
\end{example}
\section{Automated Parameter Verification and Synthesis}\label{sec:algo}

We now present our algorithms for solving the verification and approximate
synthesis problems defined in Section~\ref{sec:problem}. We first present our
algorithm for the verification problem, which will then be used as a subroutine
for the synthesis algorithm.
As a byproduct, our algorithms also produce a closed form expression for a parametric
strategy that witnesses specification satisfaction. To the best of our knowledge, this is novel with respect to the existing literature both on finite- and infinite-state pMDPs: existing works considered
the question of the \textit{existence} of concrete strategies in the verification
and synthesis problems~\cite{JungesAHJKQV24}, but how to obtain these strategies in a succinct form
was open.

Our algorithms are applicable to {\em polynomial}
pMDPs and parametric supermartingale certificates, which we formalize below. In
what follows, let $\MDP = (\States, \StatesSA, \Actions, \ActionsSA, \kernel,
\Init, \ParamSet)$ be a pMDP with $\ParamSet \subseteq \mathbb{R}^n$, $\phi$ a
specification, $p \in [0,1]$, and $\diamondsuit \in \{a,d\}$ a type of
specification satisfaction.

\subsection{Verification Algorithm}\label{sec:algoverification}

The goal of the verification problem is to formally prove that $M \models_\diamondsuit \phi,p$. Our verification algorithm proceeds by synthesizing a parametric supermartingale certificate $V$ with a supporting invariant $I$ which provides a sound proof rule for proving this statement. We aim for a unifying verification algorithm that can be used to synthesize a broad class of parametric supermartingale certificates. Thus, we first formalize the class of parametric supermartingale certificates that can be synthesized by our verification algorithm.

\smallskip\noindent{\bf Supported parametric supermartingale certificates.} Our verification algorithm can synthesize parametric generalizations of any supermartingale certificate whose defining conditions are all either {\em bound conditions} or {\em expectation conditions}. For a measurable function $V: \States \times \ParamSet \rightarrow \mathbb{R}$, we define these as follows:
\begin{compactenum}
	\item {\em Bound condition.} Consider a measurable set of states $X \in \StatesSA$ and a Borel-measurable set of real values $Y \subseteq \mathbb{R}$. We define the {\em $(X,Y)$-bound condition} to be the logical formula which encodes that, for all invariant states in the set $X$, the value of $V$ must lie in the set $Y$, i.e.
	\[ \forall u \in \ParamSet.\,\,\, \forall s \in I.\,\,\, s \in X \Rightarrow V(s,u) \in Y.  \]
	\item {\em Expectation condition.} Consider a measurable set of states $X \in \StatesSA$ and two Borel-measurable set of real values $Y,Z \subseteq \mathbb{R}$. We define the {\em $(X,Y,Z)$-expectation condition} to be the logical formula which encodes that, for all invariant states in the set $X$ at which the value of $V$ lies in the set $Y$, the difference between the value of $V$ and the expected value of $V$ upon the one-step execution of the Markov chain must lie in the set $Z$, i.e.
	\begin{equation*}
		\begin{split}
			\exists \sigma \in \StrategySpace^\MDP.\,\,\, \forall u \in \ParamSet.\,\,\, \forall s \in I.\,\,\, &s \in X \land V(s,u) \in Y \\
			&\Rightarrow V(s,u) - \mathbb{E}_{s' \sim \kernel(\cdot \mid u,s,\sigma[u](s))}[V(s')] \in Z,
		\end{split}
	\end{equation*}
	for angelic satisfaction, whereas for demonic satisfaction the existential quantification $\exists \sigma \in \StrategySpace^\MDP$ changes to the universal quantification $\forall \sigma \in \StrategySpace^\MDP$.
\end{compactenum}

Observe that the defining conditions both parametric SBFs and parametric RASMs are either bound conditions (all but the expectation decrease condition) or expectation conditions (the expectation decrease condition). This remains true for parametric generalizations of most other existing supermartingale certificates, as we discuss in Appendix~\ref{app:otherspecifications}. This includes Streett supermartingales~\cite{AbateGR24,AbateGR25} and limit-deterministic B\"uchi supermartingales~\cite{HenzingerMSZ25} for general $\omega$-regular specifications. Hence, in what follows, we present a unifying algorithmic framework for computing any parametric supermartingale certificate whose defining conditions are either bound conditions or expectation conditions.

\smallskip\noindent{\bf Assumptions.} Our algorithm is restricted to pMDPs and parametric supermartingale certificates represented in polynomial real arithmetic:
\begin{compactenum}
	\item We assume that the pMDP $\MDP$ is {\em polynomial}, meaning that:
	\begin{compactitem}
		\item The sets of parameter valuations $\ParamSet \subseteq \mathbb{R}^n$, states $\States \subseteq \mathbb{R}^{n_\States}$ and actions $\Actions \subseteq \mathbb{R}^{n_\Actions}$ are all subsets of finite-dimensional Euclidean spaces, represented as satisfiability sets of boolean combinations of finitely many polynomial inequalities over parameter variables $u_1,\dots,u_n$ (for $\ParamSet$), state variables $s_1,\dots,s_{n_\States}$ (for $\States$), and action variables $a_1,\dots,a_{n_\Actions}$ (for $\Actions$).
		\item The parametric stochastic kernel $\kernel$ can be represented in polynomial real arithmetic, that is, the probability distribution $\kernel(\cdot \mid u,s,a)$ induced by the stochastic kernel can be represented via a polynomial equation
		$ s' = \kernel(u_1,\dots,u_n,s_1,\dots,s_{n_\States},a_1,\dots,a_{n_\Actions},r_1,\dots,r_{n_r})$, 
		where $r_1,\dots,r_{n_r}$ are random variables sampled from given probability distributions that are assumed to be independent. We allow sampling from arbitrary discrete or continuous probability distributions. However, we assume that the first $D$ moments of distributions are known and available, where $D$ is the algorithm's {\em polynomial degree} parameter.
	\end{compactitem}
	\item We assume that the defining conditions of the parametric supermartingale certificate can be represented in polynomial real arithmetic. In particular, for each bound condition, we assume that $X$ and $Y$ can be represented as a satisfiability set of a boolean combination of finitely many polynomial inequalities over state variables $s_1,\dots,s_{n_\States}$ and parameters $u_1,\dots,u_n$. For each expectation condition, we assume that $X$, $Y$ and $Z$ can be represented as a satisfiability set of a boolean combination of finitely many polynomial inequalities over state variables $s_1,\dots,s_{n_\States}$ and parameters $u_1,\dots,u_n$. Note that these assumptions are clearly satisfied by parametric SBFs and RASMs.
\end{compactenum}

\smallskip\noindent{\bf Algorithm outline.} Our verification algorithm follows a template-based synthesis approach and reduces the computation of the parametric supermartingale certificate $V$ and the invariant $I$ to solving a system of polynomial real constraints. This is a standard procedure for synthesizing (non-parametric) supermartingale certificates, see e.g.~\cite{ChatterjeeFG16,ChatterjeeGMZ22,AbateGR25,HenzingerMSZ25}, hence we keep the exposition brief. In outline, the algorithm proceeds in three steps. In Step~1, the algorithm fixes a symbolic polynomial template for the parametric supermartingale certificate $V: \States \times \ParamSet \rightarrow \mathbb{R}$ and the invariant $I \subseteq \States$. In the case of angelic satisfaction, i.e.~$\diamondsuit = a$, the algorithm also fixes a template for the existentially quantified parametric strategy $\sigma: \States \times \ParamSet \rightarrow \Actions$. The templates are defined via symbolic polynomials of maximal polynomial degree $D$, which is an algorithm parameter. In Step~2, the algorithm collects a system of quantified polynomial constraints over the symbolic template variables that together entail the defining conditions of the parametric supermartingale certificate $V$ and of $I$ being a valid inductive invariant, by substituting the symbolic polynomial templates in the defining conditions. To handle quantification over the parametric strategy in the expectation condition, for the case of angelic satisfaction, i.e. $\diamond=a$, we remove the existential quantification over the parametric strategy and instead substitute the symbolic polynomial template into the parametric stochastic kernel. For the case of demonic satisfaction, i.e. $\diamond=d$, we remove the universal quantification over the parametric strategy and replace it with $\forall a \in A$, which slightly strengthens the parametric supermartingale certificate conditions and thus preserves soundness of our algorithm. 
In Step~3, this system of symbolic polynomial constraints is first reduced to solving a system of purely existentially quantified system of polynomial constraints by using existing techniques~\cite{ChatterjeeGGKSSZ25}, and finally solved by using an off-the-shelf SMT solver. Any solution gives rise to a valid instance of the parametric supermartingale certificate $V$ and the invariant $I$. We present the details of the verification algorithm and establish its soundness in Appendix~\ref{app:algoverification}.

\subsection{Synthesis Algorithm}\label{sec:algosynthesis}

We now present our approximate synthesis algorithm which, given an approximation parameter $c \in [0,1]$, partitions the parameter set $U$ into two disjoint parameter regions $R$ and $R'$ such that $M^R \models_\diamondsuit \phi,p$ and $M^{R'} \models_\clubsuit \neg\phi,1-p$ with $\clubsuit$ being the opposite type of satisfaction, where $(\mu(R) + \mu(R')) / \mu(U) \geq 1 - c$. Our synthesis algorithm uses the verification algorithm in Section~\ref{sec:algoverification} as a subroutine. To that end, we assume the subroutine $\Verify(\MDP,\phi,p,\diamondsuit)$ which either (1)~returns $\sat$ which stands for the statement $M \models_\diamondsuit \phi,p$ being formally proved, or (2)~returns $\unsat$ or $\unknown$, depending on the output of the SMT solver.

\smallskip\noindent{\bf Assumptions.} We make the same assumptions as in Section~\ref{sec:algoverification}. In addition, we assume that the parameter set $U \subseteq \mathbb{R}^n$ is bounded and covered by a hyperrectangular box $B$ (i.e.~a product of intervals). We assume that this box $B$ is provided. Note that the boundedness assumption is also necessary for the approximation synthesis problem to be well-defined, since we need $\mu(U) < \infty$.

\smallskip\noindent{\bf Algorithm.} The algorithm executes an abstraction refinement procedure which synthesizes two parameter regions $R$ and $R'$ as unions of disjoint hyperrectangular boxes. It initializes two empty sets of hyperrectangular boxes $\grid{R}$ and $\grid{R'}$, which will be gradually enlarged with verified hyperrectangular boxes until they together cover at least ratio $c$ of the total volume of $U$. We define $\volume(\grid{R}) = \sum_{b \in \grid{R}} \mu(b \cap U)$ to be the Lebesgue measure of the part of the parameter valuation set $U$ that it covers, and $\volume(\grid{R'})$ is defined analogously. Here, we assume that we have access to an oracle that can compute (or at least formally lower bound) $\mu(b \cap U)$ for any hyperrectangular box $b$. 
This assumption is made for simplicity, since volume computation is not our contribution. For the case when the parameter set $U$ is a hyperrectangle (as in our experiments), the computation is trivial and can be done exactly. For the case when $U$ is a polytope (i.e. expressible in linear real arithmetic), the computation can also be done exactly by using classical methods~\cite{beck2007computing}. Finally, for the most general case when $U$ is represented via polynomial real arithmetic, one can compute a lower bound up to arbitrary $\epsilon$-precision by using e.g. the method of~\cite{0001CFGMZ25}.

The algorithm proceeds as follows. It defines a queue $\mathtt{ToExplore}$ of hyperrectangular boxes to be explored, initialized to $\mathtt{ToExplore} = \{B\}$. It then implements the following abstraction refinement procedure:
\begin{compactenum}
	\item Check if $(\volume(\grid{R}) + \volume(\grid{R'})) / \mu(U) \geq 1 - c$. If yes, return $R$ and $R'$ to the user. Otherwise, proceed to Step~2.
	\item Pop the next box $b$ from the queue $\mathtt{ToExplore}$.
	\item Run $\Verify(\MDP^{U \cap b},\phi,p,\diamondsuit)$. If the output is $\sat$, add $b$ to $\grid{R}$.
	\item Run $\Verify(\MDP^{U \cap b},\neg\phi,1-p,\clubsuit)$. If the output is $\sat$, add $b$ to $\grid{R'}$.
	\item If neither of the two runs of $\Verify$ output $\sat$, divide $b$ along the longest side into two boxes $b_1$ and $b_2$. Add $b_1$ and $b_2$ to $\mathtt{ToExplore}$. Return to Step~1.
	
\end{compactenum}

The soundness of our approximation synthesis algorithm follows by the soundness of our verification algorithm, which ensures that at any point in the algorithm execution we have $M^R \models_\diamondsuit \phi,p$ and $M^{R'} \models_\clubsuit \neg\phi,1-p$. In addition, the termination condition ensures that $(\mu(R) + \mu(R')) / \mu(U) \geq 1 - c$ on output.

\section{Experiments}
\label{sec:experiments}

We implemented our algorithms in a prototype written in Python.\footnote{Public
URL: \url{https://github.com/kmallik/parametric-verification-and-control}}
In the back end, the quantified polynomial entailments collected in Step~2 of our verification algorithm are discharged to the third party tool PolyQEnt~\cite{chatterjee2025polyqent}. We use MathSAT5~\cite{CimattiGSS13} and Z3~\cite{MouraB08} for SMT solving. 
Experiments were run on a machine with 72 virtual cores (Xeon Gold 6154 @3GHz), 128 GB RAM (DDR4@2666 MHz).

We consider the pMDP described in Example~\ref{ex:running}, and the approximate synthesis problem for the angelic reachability (and, dually, demonic safety) specification defined in Example~\ref{ex:problem}.
We consider three different variants of the pMDP, namely (1)~with one purely additive parameter, i.e., $u_1=0$ but $u_2$ variable, written as $M^+$; (2)~with one purely multiplicative parameter, i.e., $u_2=1$ but $u_1$ variable, written as $M^\times$; and (3)~with both additive and multiplicative parameters, i.e., both $u_1$ and $u_2$ being variable, written as $M^{+,\times}$.

\begin{table}[!t]
	\vspace{-0.5cm}
	\caption{ 
	Results for different values of the approximation parameter
	$c$. Total time represents the total time the tool needed to reach the
	specified approximation level. Angelic and Demonic SAT time respective
	represent the mean and standard deviation of the time of those individual
	SMT queries that returned ``SAT.'' TO means timeout of the total runtime,
	which was set to be 900 \unit{\second}. While reporting the stats of the SMT
	calls, the remaining cases (total $-$ (SAT $+$ UNSAT $+$ inconclusive)) are
	timeouts (25 \unit{\second}).}
	\label{table:performance}
	\renewcommand{\arraystretch}{1.2} 
	\centering
	\begin{tabular}{
        |>{\centering\arraybackslash}p{1.6cm}
        |>{\centering\arraybackslash}p{0.8cm}
        |>{\centering\arraybackslash}p{0.9cm}
        |>{\centering\arraybackslash}p{0.9cm}
        |>{\centering\arraybackslash}p{0.8cm}
        |>{\centering\arraybackslash}p{0.8cm}
        |>{\centering\arraybackslash}p{0.7cm}
        |>{\centering\arraybackslash}p{0.7cm}
		|>{\centering\arraybackslash}p{1.8cm}
        |>{\centering\arraybackslash}p{1.8cm}|}
		\toprule
		&	\multirow{2}{*}{$c$}		&	\multicolumn{2}{c|}{Total time
		(\unit{\second})}	&	\multicolumn{2}{c|}{\shortstack{Av.\ time (\unit{\second})\\ ang.\ SAT
		}}	&	\multicolumn{2}{c|}{\shortstack{Av.\ time (\unit{\second})\\ dem.\ SAT
		}} &	\multicolumn{2}{c|}{\shortstack{\# SMT calls \\ Tot.-SAT-UNSAT-Inconl. }}\\
		\cline{3-10}
		&	&	Z3	&	MS	&	Z3	&	MS &	Z3	&	MS & Z3 & MS \\
		\midrule
		\multirow{5}{1.5cm}{$M^{+}$: $u_1=0$, $u_2\in [-1,1]$,
		template unknowns = 9}	&
		$0.4$	&	138.48	& 45.51		&	$ 8.16$	&	$ 4.92$	 &	$ 0.19$ & $
		0.35$ & 22k-18-22k-17 & 22-3-0-13\\
		& $0.3$	&	310.30	& 45.62		&	$ 6.04$	&	$ 2.67$ &	$ 0.19$	& $ 0.35
		 $	 & 52k-20-52k-20 & 22-4-0-14\\
		& $0.2$	&	TO	&	60.49	&	-	&	$ 2.67$ &	-	& $ 0.30
		 $	 & - & 30-5-0-19\\
		& $0.15$	&	TO	&	90.90	&	-	&	$ 1.52$ &	-	&$
		0.30$	 & - & 62-7-0-37\\
		& $0.125$	&	TO	&	TO	&	-	& -	 &	-	& -	 & - & -\\
		\midrule	
		\multirow{5}{1.5cm}{$M^{\times}$: $u_1\in [-1,1]$, $u_2=1$,
		template unknowns = 9}	& $0.4$	&
		57.47	&	75.38	&	$ 11.54$	& $ 13.69$	 &
		$ 7.11$	& $ 11.75$	 & 22-3-8-8 & 38-4-0-22 \\
		& $0.3$	&	59.14	&	89.75	&	$ 11.88$	& $ 10.90  $	 &	$ 7.11$	& $ 11.75$	 & 22-4-8-8 & 62-7-0-37\\
		& $0.2$	&	73.87	&	117.58	&	$ 11.90$	& $ 8.87  $	 &	$ 7.11$	& $ 8.68$ & 34-5-13-10	 & 154-14-0-90\\
		& $0.15$	&	83.27	&	136.08	&	$ 9.72$	& $ 9.17  $	 &	$ 6.76$	& $ 8.29$ & 46-7-18-14	 & 218-21-0-137\\
		& $0.125$	&	83.32	&	153.65	&	$ 9.72$	&$ 8.07  $	 &	$ 6.76$	& $ 8.22$ & 46-7-18-14	 & 298-31-0-209\\
		\midrule	
		\multirow{5}{1.5cm}{$M^{+,\times}$: $u_1=0$, $u_2\in [-1,1]$,
		template unknowns = 11 }	& $0.4$	&	214.01	&	76.56	&	$ 1.60
		 $	& $ 0.61$	 &	 $ 0.23$	& $ 0.57  $	 & 2k-28-1k-325 & 86-10-0-52\\
		& $0.3$	&	TO	&	91.74	&	-	& $ 0.64$	 &	-	& $ 0.61
		 $	 & - & 134-15-0-81\\
		& $0.2$	&	TO	&	122.26	&	-	&  $ 1.37$	 &	-	& $
		0.72$	 & - & 270-27-0-181\\
		& $0.15$	&	TO	&	144.18	&	-	&	$ 2.47$ &	-	& $
		0.92$	 & - & 360-42-0-256\\
		& $0.125$	&	TO	&	177.40	&	-	&	$ 3.43$ &	-	&
		$ 2.03$ & - & 530-65-0-403\\
		\bottomrule	
	\end{tabular}
	\vspace{-1em}
\end{table}

Table~\ref{table:performance} reports the results of our experiments, with
further details (variance in computation times) being available in
Table~\ref{table:performance:full} in Appendix~\ref{app:full table}. For each
of the three variants of the pMDP, we report  results for multiple values for
the approximation parameter $c$. For every call to the verification subroutine,
timeout is set to $t_{\max} = 25\, s$. 
Figure~\ref{fig:winning regions plot} visualizes the computed parameter regions. 
From the results in Table~\ref{table:performance}, we observe that $M^+$ and $M^{+,\times}$, the SMT solver MathSAT5 significantly outperforms Z3, whereas for $M^\times$, the trend is opposite.
This suggests that different SMT solvers have their own strengths and weaknesses, and choosing the right solver has significant impact on performance.
On the other hand, Figure~\ref{fig:winning regions plot} shows how our refinement algorithm is able to solve the approximate synthesis problem with reasonable accuracy.
Observe that the granularity of the cells is the finest near the frontier between the two parameter regfions, which is expected.


\begin{figure}[t]

	\begin{subcolumns}
	  \subfloat[$M^+$, $c=0.125$]{\includegraphics[scale=0.22]{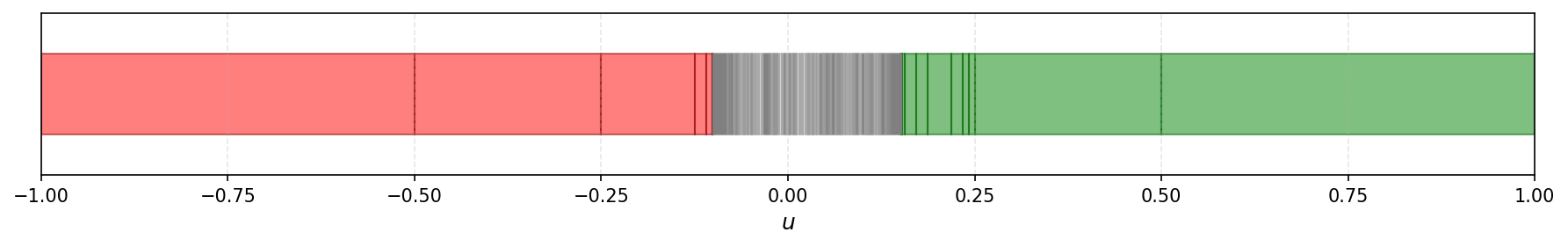}}
\nextsubfigure
  \subfloat[$M^\times$, $c=0.125$]{\includegraphics[scale=0.22]{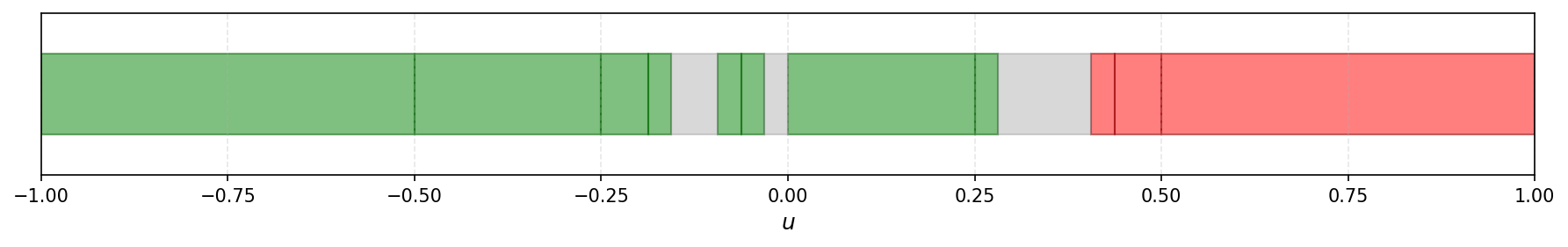}}
  \nextsubcolumn
  \qquad
  \subfloat[$M^{+,\times}$, $c=0.15$]{\includegraphics[scale=0.15]{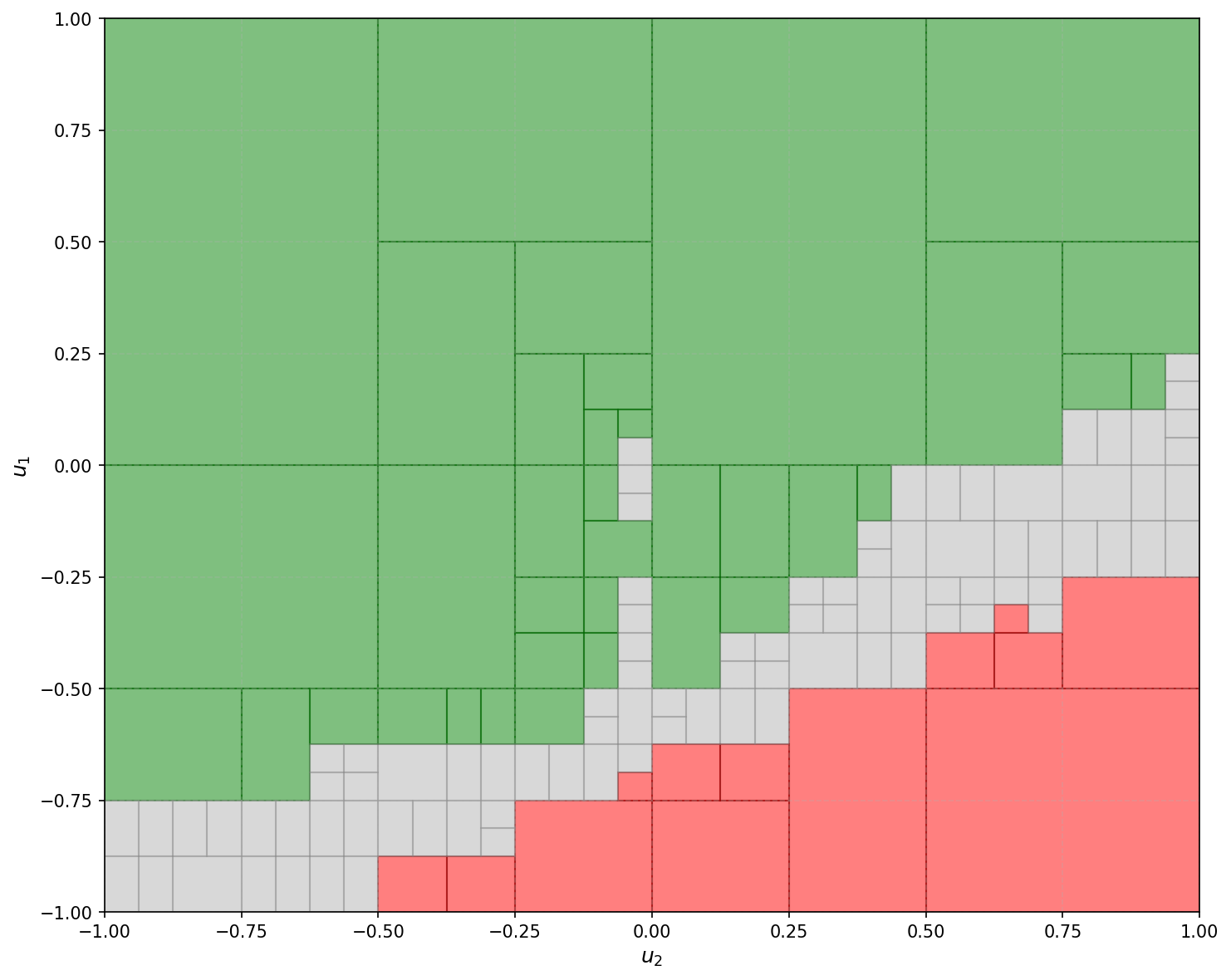}}
\end{subcolumns}
\caption{Partition of the parameter space based on the output of the approximate synthesis procedure. The green region is the angelic reachability region, the red region is the demonic safety region, and the grey region is the inconclusive region.}
\label{fig:winning regions plot}
\end{figure}

%
%
%

\section{Conclusion}

We introduced a novel automated framework for verification and approximate synthesis in infinite-state pMDPs based on parametric supermartingale certificates. By combining parameter flattening with template-based synthesis, our approach enables formal reasoning about rich specifications under parametric uncertainty.
In future, we plan to extend our methods to quantitative properties and develop heuristics to mitigate current scalability bottlenecks.

\section*{Acknowledgements} 

This research is supported by the NTU Start Up Grant and the project RYC2024-049116-I funded by MICIU/AEI/10.13039/501100011033 and the ESF+. Part of the work was done while {\DJ}or{\dj}e \v{Z}ikeli\'{c} was at Singapore Management University.

\bibliographystyle{splncs04}
\bibliography{bibliography}

\newpage
\appendix
\section{Proof of Proposition~\ref{prop:eq}}\label{app:propositionproof}

Let $\sigma$ be a parametric strategy and $\phi$ a specification in the parametric MDP, with the corresponding parameter flattened strategy $\tilde{\sigma}$ and the parameter flattened specification $\tilde{\phi}$. Let $s \in \Init$ be an initial state and $u \in U$ be a parameter valuation. To prove the proposition claim, we need to show
\[ \mathbb{P}_s^{M[u], \sigma[u]}[\phi] = \mathbb{P}_{(s,u)}^{\tilde{M},\tilde{\sigma}}[\tilde{\phi}]. \]
In what follows, denote by $(\Path^{M[u]}_s, \mathcal{F}_{\Path^{M[u]}_s}, \mathbb{P}_s^{M[u], \sigma[u]})$ the probability space induced by the MDP $M[u]$ with the initial state $s$ under strategy $\sigma[u]$, and denote by $(\Path^{\tilde{M}}_{(s,u)}, \mathcal{F}_{\Path^{\tilde{M}}_{(s,u)}}, \mathbb{P}_{(s,u)}^{\tilde{M},\tilde{\sigma}})$ the probability space induced by the parameter flattened MDP $\tilde{M}$ with the initial state $(s,u)$ under the parameter flattened strategy $\tilde{\sigma}$. Recall, the probability spaces are formally defined via the cylinder construction~\cite{ash2000probability}.

The proof proceeds in four steps:
\begin{enumerate}
	\item {\em Bijection between paths.} We construct a bijection $\Phi: \Path^{M[u]}_s \rightarrow \Path^{\tilde{M}}_{(s,u)}$ between the sets of all infinite paths in the two MDPs. The bijection is defined as follows (also defined in p.9 in the paper)
	\[ \Phi(s,a_0,s_1,a_1,\dots) = (s,u), a_0, (s_1,u), a_1, \dots \]
	with inverse
	\[ \Phi^{-1}((s,u), a_0, (s_1,u), a_1, \dots) = s,a_0,s_1,a_1,\dots. \]
	We can easily see that $\Phi^{-1} \circ \Phi$ and $\Phi \circ \Phi^{-1}$ are the identity functions over the sets of infinite paths in the two MDPs, hence, they are indeed bijections and inverses of each other.
	\item {\em Bijection between measurable sets of paths.} We now show that $\Phi$ gives rise to a bijection $\Phi^{\mathcal{F}}: \mathcal{F}_{\Path^{M[u]}_s} \rightarrow \mathcal{F}_{\Path^{\tilde{M}}_{(s,u)}}$ between the two $\sigma$-algebras of infinite paths of the two MDPs. For every measurable set of paths $B \in \mathcal{F}_{\Path^{M[u]}_s}$, define
	\[ \Phi^{\mathcal{F}}(B) = \{\Phi(\rho) \mid \rho \in B\}. \]
	We show that $\Phi^{\mathcal{F}}$ is a bijection with the inverse $(\Phi^{\mathcal{F}})^{-1}: \mathcal{F}_{\Path^{\tilde{M}}_{(s,u)}} \rightarrow \mathcal{F}_{\Path^{M[u]}_s}$ defined via
	\[ (\Phi^{\mathcal{F}})^{-1}(\tilde{B}) = \{ \Phi^{-1}(\tilde{\rho}) \mid \tilde{\rho} \in \tilde{B} \} \]
	for every $\tilde{B} \in \mathcal{F}_{\Path^{\tilde{M}}_{(s,u)}}$. The fact that $(\Phi^{\mathcal{F}})^{-1} \circ \Phi^{\mathcal{F}}$ and $\Phi^{\mathcal{F}} \circ (\Phi^{\mathcal{F}})^{-1}$ are the identity functions over the two sigma algebras easily follows from the fact that $\Phi$ and $\Phi^{-1}$ are inverses. Hence, $\Phi^{\mathcal{F}}$ and $(\Phi^{\mathcal{F}})^{-1}$ are indeed bijections and inverses of each other.
	\item {\em Bijection is measure preserving.} Next, we show that the bijection $\Phi^{\mathcal{F}}$ over $\sigma$-algebras is measure preserving, i.e. that
	\[ \mathbb{P}_s^{M[u], \sigma[u]}[B] = \mathbb{P}_{(s,u)}^{\tilde{M},\tilde{\sigma}}[\Phi^{\mathcal{F}}(B)] \]
	holds for all $B \in \mathcal{F}_{\Path^{M[u]}_s}$. To prove this, observe first that $\mathbb{P}_s^{M[u], \sigma[u]}[\cdot]$ and $\mathbb{P}_{(s,u)}^{\tilde{M},\tilde{\sigma}}[\Phi^{\mathcal{F}}(\cdot)]$ are two probability measures over the measurable spaces $(\Path^{M[u]}_s, \mathcal{F}_{\Path^{M[u]}_s})$ and $(\Path^{\tilde{M}}_{(s,u)}, \mathcal{F}_{\Path^{\tilde{M}}_{(s,u)}})$, respectively. Thus, we want to prove that these two probability measures are equal on all $B \in \mathcal{F}_{\Path^{M[u]}_s}$. To prove this, it suffices to prove that they are equal on a set of events that generates the $\sigma$-algebra $\mathcal{F}_{\Path^{M[u]}_s}$ (see e.g.~\cite[Section~1.6]{williams1991probability}, where it is shown that any two probability measures that agree on a $\pi$-system must agree on the whole $\sigma$-algebra generated by that $\pi$-system). In our case, the $\sigma$-algebra $\mathcal{F}_{\Path^{M[u]}_s}$ is generated by cylinder sets, i.e. sets of the form $Cyl(s,a_0,s_1,a_1,\dots,s_k) = \{\rho \in Path^{M[u]}_s \mid s,a_0,s_1,a_1,\dots,s_k \text{ is a prefix of } \rho \}$. Now, observe that 
	\begin{equation*}
		\begin{split}
			&\mathbb{P}_s^{M[u], \sigma[u]}[Cyl(s,a_0,s_1,a_1,\dots,s_k)] \\
			&= \mathbb{P}_{(s,u)}^{\tilde{M},\tilde{\sigma}}[Cyl((s,u),a_0,(s_1,u),a_1,\dots,(s_k,u)))] \\
			&= \mathbb{P}_{(s,u)}^{\tilde{M},\tilde{\sigma}}[\Phi^{\mathcal{F}}(Cyl(s,a_0,s_1,a_1,\dots,s_k))], 
		\end{split}
	\end{equation*}
	where the first equality follows from the cylinder construction and how the probability measures over the two MDPs are defined, whereas the second equality follows from our definition of the bijection $\Phi^{\mathcal{F}}$. Hence, the two probability measures agree on all cylinder sets which generate the $\sigma$-algebra $\mathcal{F}_{\Path^{M[u]}_s}$, hence they also must agree on all elements of the $\sigma$-algebra $\mathcal{F}_{\Path^{M[u]}_s}$. This concludes our proof that $\Phi^{\mathcal{F}}$ is measure preserving.
	
	\item {\em Proof of the proposition claim.} We are finally ready to prove the proposition claim, i.e. that $\mathbb{P}_s^{M[u], \sigma[u]}[\phi] = \mathbb{P}_{(s,u)}^{\tilde{M},\tilde{\sigma}}[\tilde{\phi}]$. First, recall that the parameter flattened specification $\tilde{\phi}$ is defined via
	\[ \tilde{\phi} = \{(s,u),a_0,(s_1,u),a_1,\dots \mid s,a_0,s_1,a_1,\dots \in \phi \land u \in U \}.  \]
	Hence, we have $\tilde{\phi} = \Phi^{\mathcal{F}}(\phi)$. Now, by Step 3 above $\Phi^{\mathcal{F}}$ is measure preserving, so it follows that $\mathbb{P}_s^{M[u], \sigma[u]}[\phi] = \mathbb{P}_{(s,u)}^{\tilde{M},\tilde{\sigma}}[\tilde{\phi}]$, which proves the claim.
\end{enumerate}
\qed

\section{Proof of Theorem~\ref{thm:transformation}}\label{app:transformation}

We distinguish between the cases of angelic and demonic satisfaction.

\smallskip\noindent{\em Angelic satisfaction.}  Suppose that $\diamondsuit = a$. We have
\[ \MDP \models_a \phi,p \quad \equiv \quad \exists \sigma \in \StrategySpace^\MDP.\,\,\, \forall u \in U.\,\,\, \forall s \in \Init.\,\,\, \mathbb{P}_s^{\MDP[u],\sigma[u]}[\phi] \geq p. \]
On the other hand, for the flattened MDP $\tilde{\MDP}$, the parameter valuations are now incorporated into the state space so angelic satisfaction reduces to
\[ \tilde{\MDP} \models_a \tilde{\phi},p \quad \equiv \quad \exists \sigma \in \StrategySpace^{\tilde{\MDP}}.\,\,\, \forall (s,u) \in \Init \times \ParamSet.\,\,\, \mathbb{P}_{(s,u)}^{\tilde{\MDP},\tilde{\sigma}}[\tilde{\phi}] \geq p. \]
Since each parametric strategy in $\MDP$ induces a parameter flattened strategy $\tilde{\sigma}$ in the parameter flattened MDP $\tilde{\MDP}$ and vice versa, and since by Proposition~\ref{prop:eq} we have that $\mathbb{P}_s^{\MDP[u],\sigma[u]}[\phi] \geq p$ if and only if $\mathbb{P}_{(s,u)}^{\tilde{\MDP},\tilde{\sigma}}[\tilde{\phi}] \geq p$, we conclude that $\MDP \models_a \phi,p$ if and only if $\tilde{\MDP} \models_a \tilde{\phi},p$.

\smallskip\noindent{\em Demonic satisfaction.} Suppose that $\diamondsuit = d$. We have
\[ \MDP \models_d \phi,p \quad \equiv \quad \forall \sigma \in \StrategySpace^\MDP.\,\,\, \forall u \in U.\,\,\, \forall s \in \Init.\,\,\, \mathbb{P}_s^{\MDP[u],\sigma[u]}[\phi] \geq p. \]
On the other hand, for the flattened MDP $\tilde{\MDP}$, the parameter valuations are now incorporated into the state space so demonic satisfaction reduces to
\[ \tilde{\MDP} \models_d \tilde{\phi},p \quad \equiv \quad \forall \sigma \in \StrategySpace^{\tilde{\MDP}}.\,\,\, \forall (s,u) \in \Init \times \ParamSet.\,\,\, \mathbb{P}_{(s,u)}^{\tilde{\MDP},\tilde{\sigma}}[\tilde{\phi}] \geq p. \]
Since each parametric strategy in $\MDP$ induces a parameter flattened strategy $\tilde{\sigma}$ in the parameter flattened MDP $\tilde{\MDP}$ and vice versa, and since by Proposition~\ref{prop:eq} we have that $\mathbb{P}_s^{\MDP[u],\sigma[u]}[\phi] \geq p$ if and only if $\mathbb{P}_{(s,u)}^{\tilde{\MDP},\tilde{\sigma}}[\tilde{\phi}] \geq p$, we conclude that $\MDP \models_d \phi,p$ if and only if $\tilde{\MDP} \models_d \tilde{\phi},p$. \qed

\section{Proof of Theorem~\ref{thm:psbf}}\label{app:psbf}

Suppose that there exists a parametric SBF $V$ for $X$ with respect to $I$ and $\diamondsuit$. We distinguish between angelic and demonic satisfaction.

\smallskip\noindent{\em Angelic satisfaction.} Consider first the case of~$\diamondsuit = a$. Let $\sigma$ be a parametric strategy whose existence is ensured by the Expected decrease condition for angelic satisfaction, and let $\tilde{\sigma}$ be the parameter flattened strategy of $\sigma$. Then, by the defining conditions of parametric SBFs in Definition~\ref{def:psbf}, it follows that $V: \States \times \ParamSet \rightarrow \mathbb{R}$ is an SBF for the safe set $X \times \ParamSet$ with respect to the invariant $I \times \ParamSet$ in the non-parametric Markov chain defined by the parameter flattened MDP $\tilde{\MDP}$ and the parameter flattened strategy $\tilde{\sigma}$. Hence, by the soundness of SBFs for proving probability~$p \in [0,1)$ safety in Markov chains~\cite{PrajnaJP07}, it follows that for all initial states $(s,u) \in \Init \times \ParamSet$ in the parameter flattened MDP $\tilde{\MDP}$, we have $\mathbb{P}_{(s,u)}^{\tilde{\MDP},\tilde{\sigma}}[\Safe(X \times \ParamSet)] \geq p$. Thus, we have $\tilde{\MDP} \models_a \Safe(X \times \ParamSet), p$. Hence, by Theorem~\ref{thm:transformation}, we conclude that $\MDP \models_a \Safe(X), p$, as claimed.

\smallskip\noindent{\em Demonic satisfaction.} Consider now the case of~$\diamondsuit = d$. Let $\tilde{\sigma}$ be an arbitrary strategy in the parameter flattened MDP $\tilde{\MDP}$. Then, by the defining conditions of parametric SBFs in Definition~\ref{def:psbf}, it follows that $V: \States \times \ParamSet \rightarrow \mathbb{R}$ is an SBF for the safe set $X \times \ParamSet$ with respect to the invariant $I \times \ParamSet$ in the non-parametric Markov chain defined by the parameter flattened MDP $\tilde{\MDP}$ and the parameter flattened strategy $\tilde{\sigma}$. Hence, by the soundness of SBFs for proving probability $p \in [0,1)$ safety in Markov chains~\cite{PrajnaJP07}, it follows that for all initial states $(s,u) \in \Init \times \ParamSet$ in the parameter flattened MDP $\tilde{\MDP}$, we have $\mathbb{P}_{(s,u)}^{\tilde{\MDP},\tilde{\sigma}}[\Safe(X \times \ParamSet)] \geq p$. Since the strategy $\tilde{\sigma}$ in $\tilde{\MDP}$ was arbitrary, we have $\tilde{\MDP} \models_d \Safe(X \times \ParamSet), p$. Hence, by Theorem~\ref{thm:transformation}, we conclude that $\MDP \models_d \Safe(X), p$, as claimed. \qed

\section{Proof of Theorem~\ref{thm:prasm}}\label{app:prasm}

Suppose that there exists a parametric RASM $V$ for $T$ with respect to $I$ and $\diamondsuit$. We distinguish between angelic and demonic satisfaction.

\smallskip\noindent{\em Angelic satisfaction.} Consider first the case of~$\diamondsuit = a$. Let $\sigma$ be a parametric strategy whose existence is ensured by the Expected decrease condition for angelic satisfaction, and let $\tilde{\sigma}$ be the parameter flattened strategy of $\sigma$. Then, by the defining conditions of parametric RASMs in Definition~\ref{def:prasm}, it follows that $V: \States \times \ParamSet \rightarrow \mathbb{R}$ is a RASM for the target set $T \times \ParamSet$ with respect to the invariant $I \times \ParamSet$ in the non-parametric Markov chain defined by the parameter flattened MDP $\tilde{\MDP}$ and the parameter flattened strategy $\tilde{\sigma}$. Hence, by the soundness of RASMs for proving probability~$p \in [0,1)$ reachability in Markov chains~\cite{ZikelicLHC23}, it follows that for all initial states $(s,u) \in \Init \times \ParamSet$ in the parameter flattened MDP $\tilde{\MDP}$, we have $\mathbb{P}_{(s,u)}^{\tilde{\MDP},\tilde{\sigma}}[\Reach(T \times \ParamSet)] \geq p$. Thus, we have $\tilde{\MDP} \models_a \Reach(T \times \ParamSet), p$. Hence, by Theorem~\ref{thm:transformation}, we conclude that $\MDP \models_a \Reach(T), p$, as claimed.

\smallskip\noindent{\em Demonic satisfaction.} Consider now the case of~$\diamondsuit = d$. Let $\tilde{\sigma}$ be an arbitrary strategy in the parameter flattened MDP $\tilde{\MDP}$. Then, by the defining conditions of parametric RASMs in Definition~\ref{def:prasm}, it follows that $V: \States \times \ParamSet \rightarrow \mathbb{R}$ is a RASM for the target set $T \times \ParamSet$ with respect to the invariant $I \times \ParamSet$ in the non-parametric Markov chain defined by the parameter flattened MDP $\tilde{\MDP}$ and the parameter flattened strategy $\tilde{\sigma}$. Hence, by the soundness of RASMs for proving probability $p \in [0,1)$ reachability in Markov chains~\cite{ZikelicLHC23}, it follows that for all initial states $(s,u) \in \Init \times \ParamSet$ in the parameter flattened MDP $\tilde{\MDP}$, we have $\mathbb{P}_{(s,u)}^{\tilde{\MDP},\tilde{\sigma}}[\Reach(T \times \ParamSet)] \geq p$. Since the strategy $\tilde{\sigma}$ in $\tilde{\MDP}$ was arbitrary, we have $\tilde{\MDP} \models_d \Reach(T \times \ParamSet), p$. Hence, by Theorem~\ref{thm:transformation}, we conclude that $\MDP \models_d \Reach(T), p$, as claimed. \qed

\section{Parametric Supermartingales for Probability~1 Reachability}\label{app:parametricrsm}

We now present a parametric spermartingale certificate generalization of {\em ranking supermartingaes (RSMs)} for proving probability~$1$ (a.k.a.~almost-sure) reachability~\cite{ChakarovS13}. We first recall (non-parametric) RSMs. Intuitively, given a Markov chain $\MDP$, a set of target states $T$ that we want to prove is reached with probability~$1$, and an invariant $I$ that over-approximates the set of all reachable states, an RSM is a measurable function $V: \States \rightarrow \mathbb{R}$ that to each state assigns a real value that is required to satisfy the following two conditions: {\bf (C1)}~$V$ is non-negative at all states in the invariant $I$, and {\bf (C2)}~at all non-target states, i.e.~states in $I \backslash T$, $V$ strictly decreases in expected value by at least $\epsilon > 0$ upon the one-step execution of the Markov chain.

The following definition formalizes the above intuition and generalizes it to the setting of parametric MDPs by instantiating it in the state space of parameter flattened MDP. Since parametric RSMs are defined over the state space of parameter flattened MDP, they are defined as measurable functions of type $V: \States \times \ParamSet \rightarrow \mathbb{R}$ that map {\em states of the parameter flattened MDP} to real values.

\begin{definition}[Parametric RSMs]\label{def:prsm}
	Let $\MDP = (\States, \StatesSA, \Actions, \ActionsSA, \kernel, \Init, \ParamSet)$ be a parametric MDP, $I \in \StatesSA$ be an invariant, $T \in \StatesSA$ be a set of target states, and $\diamondsuit \in \{a,d\}$ denote either angelic or demonic satisfaction. A {\em parametric ranking supermartingale (parametric RSM)} for $T$ with respect to $I$ and $\diamondsuit$ is a measurable function $V: \States \times \ParamSet \rightarrow \mathbb{R}$, which is required to satisfy the following four conditions:
	\begin{compactenum}
		\item {\em Nonnegativity.} $V(s,u) \geq 0$ for all $s \in I$ and $u \in U$.
		\item {\em Expected decrease.} We distinguish between angelic and demonic satisfaction:
		\begin{compactenum}
			\item {\em Angelic satisfaction, i.e.~$\diamondsuit = a$.} There exist $\epsilon > 0$ and a parametric strategy $\sigma: \ParamSet \times \States \rightarrow \Actions$ in $\MDP$ such that, for every parameter valuation $u \in U$ and for every non-target state $s \in I \backslash T$, $V$ satisfies expected decrease for the action $a = \sigma[u](s)$, i.e.
			\begin{equation*}
				\begin{split}
					\exists \strategy \in \StrategySpace^\MDP.\,\,\,  \forall u \in \ParamSet.\,\,\, &\forall s \in I.\,\,\, s \in I \backslash T \\
					&\Rightarrow V(s,u)  \geq \mathbb{E}_{s' \sim \kernel(\cdot \mid u,s,\sigma[u](s))}[V(s',u)] + \epsilon.
				\end{split}
			\end{equation*}
			\item {\em Demonic satisfaction, i.e.~$\diamondsuit = d$.} There exists $\epsilon > 0$ such that for all parametric strategies $\sigma: \ParamSet \times \States \rightarrow \Actions$ in $\MDP$, for every parameter valuation $u \in U$ and for every non-target state $s \in I \backslash T$, $V$ satisfies expected decrease for the action $a = \sigma[u](s)$, i.e.
			\begin{equation*}
				\begin{split}
					\forall \strategy \in \StrategySpace^\MDP.\,\,\,  \forall u \in \ParamSet.\,\,\, &\forall s \in I.\,\,\, s \in I \backslash T \\
					&\Rightarrow V(s,u)  \geq \mathbb{E}_{s' \sim \kernel(\cdot \mid u,s,\sigma[u](s))}[V(s',u)] + \epsilon.
				\end{split}
			\end{equation*}
		\end{compactenum}
	\end{compactenum}
\end{definition}

The following theorem shows that parametric RSMs provide a sound proof rule for proving probability~$1$ reachability in parametric MDPs. 

\begin{theorem}[Soundness of parametric RSMs]\label{thm:prsm}
	Let $\MDP$ be a parametric MDP, $I$ be an invariant, $T$ be a set of target states, and $\diamondsuit \in \{a,d\}$ denote either angelic or demonic satisfaction. Suppose that there exists a parametric RSM for $T$ with respect to $I$ and $\diamondsuit$. Then, we have $\MDP \models_\diamondsuit \Reach(T), 1$.
\end{theorem}

\begin{proof}
	Suppose that there exists a parametric RSM $V$ for $T$ with respect to $I$ and $\diamondsuit$. We distinguish between angelic and demonic satisfaction.
	
	\smallskip\noindent{\em Angelic satisfaction.} Consider first the case of~$\diamondsuit = a$. Let $\sigma$ be a parametric strategy whose existence is ensured by the Expected decrease condition for angelic satisfaction, and let $\tilde{\sigma}$ be the parameter flattened strategy of $\sigma$. Then, by the defining conditions of parametric RSMs in Definition~\ref{def:prsm}, it follows that $V: \States \times \ParamSet \rightarrow \mathbb{R}$ is an RSM for the target set $T \times \ParamSet$ with respect to the invariant $I \times \ParamSet$ in the non-parametric Markov chain defined by the parameter flattened MDP $\tilde{\MDP}$ and the parameter flattened strategy $\tilde{\sigma}$. Hence, by the soundness of RSMs for proving probability~1 reachability in Markov chains~\cite{ChakarovS13}, it follows that for all initial states $(s,u) \in \Init \times \ParamSet$ in the parameter flattened MDP $\tilde{\MDP}$, we have $\mathbb{P}_{(s,u)}^{\tilde{\MDP},\tilde{\sigma}}[\Reach(T \times \ParamSet)] = 1$. Thus, we have $\tilde{\MDP} \models_a \Reach(T \times \ParamSet), 1$. Hence, by Theorem~\ref{thm:transformation}, we conclude that $\MDP \models_a \Reach(T), 1$, as claimed.
	
	\smallskip\noindent{\em Demonic satisfaction.} Consider now the case of~$\diamondsuit = d$. Let $\tilde{\sigma}$ be an arbitrary strategy in the parameter flattened MDP $\tilde{\MDP}$. Then, by the defining conditions of parametric RSMs in Definition~\ref{def:prsm}, it follows that $V: \States \times \ParamSet \rightarrow \mathbb{R}$ is an RSM for the target set $T \times \ParamSet$ with respect to the invariant $I \times \ParamSet$ in the non-parametric Markov chain defined by the parameter flattened MDP $\tilde{\MDP}$ and the parameter flattened strategy $\tilde{\sigma}$. Hence, by the soundness of RSMs for proving probability~1 reachability in Markov chains~\cite{ChakarovS13}, it follows that for all initial states $(s,u) \in \Init \times \ParamSet$ in the parameter flattened MDP $\tilde{\MDP}$, we have $\mathbb{P}_{(s,u)}^{\tilde{\MDP},\tilde{\sigma}}[\Reach(T \times \ParamSet)] = 1$. Since the strategy $\tilde{\sigma}$ in $\tilde{\MDP}$ was arbitrary, we have $\tilde{\MDP} \models_d \Reach(T \times \ParamSet), 1$. Hence, by Theorem~\ref{thm:transformation}, we conclude that $\MDP \models_d \Reach(T), 1$, as claimed. \qed
\end{proof}

\section{Parametric Generalization of Other Supermartingale Certificates}\label{app:otherspecifications}

Our presentation so far has focused on parametric generalizations of established supermartingale certificates, such as RSMs, SBFs and RASMs, which serve as our illustrating examples throughout the paper. However, our generalization is applicable to all other supermartingale certificates for measurable Markov chains whose defining conditions are of one of the two types that we formalize below.

\smallskip\noindent{\bf Bound and expectation conditions.} Let $\MDP = (\States, \StatesSA, \Actions, \ActionsSA, \kernel, \Init)$ be a Markov chain and $I \in \StatesSA$ be an invariant. For a measurable function $V: \States \rightarrow \mathbb{R}$, we define bound conditions and expectation conditions as follows:
\begin{compactenum}
	\item {\em Bound condition.} Consider a measurable set of states $X \in \StatesSA$ and a Borel-measurable set of real values $Y \subseteq \mathbb{R}$. We define the {\em $(X,Y)$-bound condition} to be the logical formula which encodes that, for all invariant states in the set $X$, the value of $V$ must lie in the set $Y$, i.e.
	\[ \forall s \in I.\,\,\, s \in X \Rightarrow V(s) \in Y.  \]
	\item {\em Expectation condition.} Consider a measurable set of states $X \in \StatesSA$ and two Borel-measurable set of real values $Y,Z \subseteq \mathbb{R}$. We define the {\em $(X,Y,Z)$-expectation condition} to be the logical formula which encodes that, for all invariant states in the set $X$ at which the value of $V$ lies in the set $Y$, the difference between the value of $V$ and the expected value of $V$ upon the one-step execution of the Markov chain must lie in the set $Z$, i.e.
	\[ \forall s \in I.\,\,\, s \in X \land V(s) \in Y \Rightarrow V(s) - \mathbb{E}_{s'\in\kernel(\cdot \mid s)} \in Z.  \]
\end{compactenum}

Observe that the defining conditions of all supermartingale certificates considered in this work, namely RSMs, SBFs and RASMs, are either bound conditions (all but the expectation decrease condition) or expectation conditions (the expectation decrease condition). This remains true for many other supermartingale certificates considered in the literature as well. For instance, all defining conditions of repulsing supermartingales~\cite{ChatterjeeNZ17} which were also considered for proving quantiative reachability and safety, are either bound or expectation conditions. This is also true even the recently introduced Streett supermartingales~\cite{AbateGR24,AbateGR25} and limit-deterministic B\"uchi supermartingales~\cite{HenzingerMSZ25} for proving satisfaction of general $\omega$-regular specifications.

\smallskip\noindent{\bf Generalization to parametric MDPs.} Consider now a supermartingale certificate for measurable Markov chains whose all defining conditions are either bound or expectation conditions. We extend it to the setting of parametric MDPs as follows. Let $\MDP = (\States, \StatesSA, \Actions, \ActionsSA, \kernel, \Init, \ParamSet)$ be a parametric MDP, $I \in \StatesSA$ be an invariant, and $\diamondsuit \in \{a,d\}$ denote either angelic or demonic satisfaction. For each defining condition of the supermartingale certificate in Markov chains, we define its {\em parametric generalization} as follows:
\begin{compactenum}
	\item {\em Parametric generalization of bound conditions.} For a bound condition $\forall s \in I.\,\,\, s \in X \Rightarrow V(s) \in Y$, its parametric generalization remains unchanged, i.e.~we keep the condition $\forall s \in I.\,\,\, s \in X \Rightarrow V(s) \in Y$.
	\item {\em Parametric generalization of expectation conditions.} For an expectation condition $\forall s \in I.\,\,\, s \in X \land V(s) \in Y \Rightarrow V(s) - \mathbb{E}_{s'\in\kernel(\cdot \mid s)} \in Z$, we generalize it to the setting of parametric MDPs as follows depending on whether we are interested in angelic satisfaction or demonic satisfaction:
	\begin{compactenum}
		\item {\em Angelic satisfaction.} If $\diamondsuit = a$, then the expectation condition is generalized to the setting of parametric MDPs as follows: We require that there exists a parametric strategy $\sigma: \ParamSet \times \States \rightarrow \Actions$ such that, for every parameter valuation $u \in \ParamSet$ and for every state $s \in I$, if the premise of the entailment holds then the conclusion of the entailment holds, i.e.
		\begin{equation*}
		\begin{split}
			\exists \strategy \in \StrategySpace^{\MDP}.\,\,\, \forall u\in\ParamSet.\,\,\, \forall s \in I.\,\,\, &s \in X \land V(s) \in Y \\
			&\Rightarrow V(s) - \mathbb{E}_{s'\in\kernel(\cdot \mid u,s,\sigma[u](s))} \in Z
		\end{split}
		\end{equation*}
		\item {\em Demonic satisfaction.} If $\diamondsuit = d$, then the expectation condition is generalized to the setting of parametric MDPs as follows: We require that for all parametric strategies $\sigma: \ParamSet \times \States \rightarrow \Actions$, for every parameter valuation $u \in \ParamSet$ and for every state $s \in I$, if the premise of the entailment holds then the conclusion of the entailment holds, i.e.
		\begin{equation*}
			\begin{split}
				\forall \strategy \in \StrategySpace^{\MDP}.\,\,\, \forall u\in\ParamSet.\,\,\, \forall s \in I.\,\,\, &s \in X \land V(s) \in Y \\
				&\Rightarrow V(s) - \mathbb{E}_{s'\in\kernel(\cdot \mid u,s,\sigma[u](s))} \in Z
			\end{split}
		\end{equation*}
	\end{compactenum}
\end{compactenum}
Note that the defining conditions of parametric RSMs in Definition~\ref{def:prsm}, parametric SBFs in Definition~\ref{def:psbf} and parametric RASMs in Definition~\ref{def:prasm} are obtained precisely by applying these two parametric generalizations to the defining conditions of RSMs~\cite{ChakarovS13}, SBFs~\cite{PrajnaJP07} and RASMs~\cite{ZikelicLHC23}. It can also be analogously applied to obtain the parametric generalization of any other supermartingale certificate for Markov chains whose defining conditions are all either bound or expectation conditions. The proof of generalization soundness follows analogously as the proofs of Theorem~\ref{thm:prsm}, Theorem~\ref{thm:psbf} and Theorem~\ref{thm:prasm}, which all proceed by showing that the parametric supermartingale certificate induces a non-parametric supermartingale certificate in the parameter flattened MDP and then using the soundness of the non-parametric supermartingale certificate for proving the specification of interest, while taking additional care to correctly handle the cases of angelic and demonic satisfaction.

\section{Verification Algorithm}\label{app:algoverification}

Our verification algorithm follows a template-based synthesis approach and reduces the computation of the parametric supermartingale certificate $V$ and the invariant $I$ to solving a system of polynomial real constraints. This is a standard procedure for synthesizing (non-parametric) supermartingale certificates, see e.g.~\cite{ChatterjeeFG16,ChatterjeeGMZ22,AbateGR25,HenzingerMSZ25}. The algorithm proceeds in three steps.

\smallskip\noindent{\bf Step 1: Fixing templates.} The algorithm fixes symbolic polynomial templates for the parametric supermartingale certificate $V: \States \times \ParamSet \rightarrow \mathbb{R}$ and the invariant $I \subseteq \States$. In the case of angelic satisfaction, i.e.~$\diamondsuit = a$, the algorithm also fixes a template for the existentially quantified parametric strategy $\sigma: \States \times \ParamSet \rightarrow \Actions$.

For the parametric supermartingale certificate $V$, the symbolic polynomial template is defined by a polynomial expression $P_V(s_1,\dots,s_{n_\States},u_1,\dots,u_n)$ over the state variables and parameters of polynomial degree $D$, which is an algorithm parameter. The values of coefficients of each monomial are at this stage unknown and are given by the {\em symbolic template variables}, whose concrete values will be computed at the later stages of the algorithm.

For the invariant $I$, since it is a subset of the state space, we define its symbolic polynomial template to be the satisfiability set of a conjunction of finitely many symbolic polynomial inequalities $(P^1_I(s)) \geq 0 \land \dots \land (P^{N_I}_I(s) \geq 0)$, where each $P^j_I(s_1,\dots,s_{n_\States})$ is a polynomial over the state variables of polynomial degree $D$, while the number $N_I$ of inequalities is another algorithm parameter.

Finally, if $\diamondsuit = a$, the symbolic polynomial template for the existentially quantified parametric strategy $\sigma$ is a symbolic polynomial $P_\sigma(s_1,\dots,s_{n_\States},u_1,\dots,u_n)$ over the state variables and parameters of polynomial degree $D$.

\begin{example}
	Suppose that the dimension of the state space is $p=1$ and the number of parameters is $n=2$, as in Example~\ref{ex:running}. Suppose that the polynomial degree parameter is $D=1$ and $N_I = 1$. The symbolic polynomial template of $V$ is a degree $1$ polynomial expression defined by $P_V(s_1,u_1,u_2) = \alpha_1 \cdot s_1 + \alpha_2 \cdot u_1 + \alpha_3 \cdot u_2 + \alpha_4$, where $\alpha_1,\dots,\alpha_4$ are the symbolic template variables. The symbolic polynomial template for $I$ is $\{s \in \States \mid P_I^1(s_1) \geq 0\}$ with $P_I^1(s_1) = \beta_1 \cdot s_1 + \beta_2$, where $\beta_1, \beta_2$ are the symbolic template variables. If $\diamondsuit = a$, the symbolic polynomial template for $\sigma$ is $P_\sigma(s_1,u_1,u_2) = \gamma_1 \cdot s_1 + \gamma_2 \cdot u_1 + \gamma_3 \cdot u_2 + \gamma_4$, where $\gamma_1,\dots,\gamma_4$ are the symbolic template variables.
\end{example}

\smallskip\noindent{\bf Step 2: Constraint collection.} The algorithm collects a system of polynomial constraints over the symbolic template variables which together encode that all defining conditions of the parametric supermartingale certificate $V$ are satisfied, as well as that the invariant $I$ is an inductive invariant. This is achieved by substituting the polynomial expressions that define the symbolic polynomial templates in Step~1 into the defining conditions of parametric supermartingale certificates and inductive invariants. For the case of angelic satisfaction, i.e.~$\diamondsuit = a$, we remove the existential quantification over the parametric strategy and instead substitute the symbolic polynomial template $P_\sigma(s,u)$ into the parametric stochastic kernel. For the case of demonic satisfaction, i.e.~$\diamondsuit = d$, we remove the universal quantification over the parametric strategy and replace it with $\forall a \in \Actions$, which slightly strengthens the parametric supermartingale certificate conditions and thus preserves soundness of our algorithm. Since constraint collection is standard and analogous to existing template-based synthesis algorithms for non-parametric supermartingale certificates~\cite{ChatterjeeFG16,ChatterjeeGMZ22,AbateGR25,HenzingerMSZ25}, we omit the details.

\smallskip\noindent{\bf Step 3: Constraint solving.} The previous step results in a system of {\em quantified polynomial entailments}, i.e.~constraints of the form
\[ \forall a \in \Actions.\,\,\, \forall u \in \ParamSet.\,\,\, \forall s \in I.\,\,\, \Phi(s,u) \Longrightarrow \Psi(s,u,a), \]
where $\Phi$ and $\Psi$ are boolean combinations of symbolic polynomial inequalities over the state, parameter and action variables, whose coefficients are either concrete real values or symbolic template variables introduced in Step~1. Our algorithm now solves this system of constraints as follows. First, by using Putinar's theorem~\cite{putinar1993positive} (or Handelman's theorem~\cite{handelman1988representing}, if all polynomials are of degree $1$), each quantified polynomial entailment is translated into a system of purely existentially quantified polynomial constraints over the symbolic template variables, as well as fresh variables introduced by the Putinar’s theorem. This step is again standard so we omit the details and refer the reader to~\cite{ChatterjeeFG16,ChatterjeeGMZ22,AbateGR25,HenzingerMSZ25}. The resulting system is then solved by an SMT solver. This translation and reduction to SMT solving can be achieved via an off-the-shelf PolyQEnt tool~\cite{ChatterjeeGGKSSZ25}, as we do in our implementation, which also includes details on how boolean combinations of polynomial inequalities are handled before applying Putinar's theorem.

The following theorem establishes soundness of our algorithm. The proof is immediate since the system of constraints collected in Step~2 entails that $V$ is a valid parametric supermartingale certificate with the supporting invariant $I$, and the reduction to SMT solving in Step~3 is sound.

\begin{theorem}[Soundness]\label{thm:verification}
	Suppose that the verification algorithm returns a parametric supermartingale certificate $V$ and an invariant $I$. Then $M \models_\diamondsuit \phi,p$.
\end{theorem}

\section{Full Experimental Results Table (including Variance in Runtime)}
\label{app:full table}
\begin{table}[!t]
	\vspace{-0.5cm}
	\caption{\textit{Extended version of Table~\ref{table:performance} with variance of
	computation times.} 
	Results for different values of the approximation parameter
	$c$. Total time represents the total time the tool needed to reach the
	specified approximation level. Angelic and Demonic SAT time respective
	represent the mean and standard deviation of the time of those individual
	SMT queries that returned ``SAT.'' TO means timeout of the total runtime,
	which was set to be 900 \unit{\second}. While reporting the stats of the SMT
	calls, the remaining cases (total $-$ (SAT $+$ UNSAT $+$ inconclusive)) are
	timeouts (25 \unit{\second}).}
	\label{table:performance:full}
	\renewcommand{\arraystretch}{1.5} 
	\centering
	\Rotatebox{90}{%
	\begin{tabular}{
        |>{\centering\arraybackslash}p{2.5cm}
        |>{\centering\arraybackslash}p{1cm}
        |>{\centering\arraybackslash}p{1.6cm}
        |>{\centering\arraybackslash}p{1.6cm}
        |>{\centering\arraybackslash}p{1.6cm}
        |>{\centering\arraybackslash}p{1.6cm}
        |>{\centering\arraybackslash}p{1.6cm}
        |>{\centering\arraybackslash}p{1.6cm}
		|>{\centering\arraybackslash}p{2.7cm}
        |>{\centering\arraybackslash}p{2.2cm}|}
		\toprule
		&	\multirow{2}{*}{$c$}		&	\multicolumn{2}{c|}{Total time
		(\unit{\second})}	&	\multicolumn{2}{c|}{Angelic SAT time
		(\unit{\second})}	&	\multicolumn{2}{c|}{Demonic SAT time
		(\unit{\second})} &	\multicolumn{2}{c|}{\shortstack{\# SMT calls and
		their outcome \\ total-SAT-UNSAT-inconclusive }}\\
		\cline{3-10}
		&	&	Z3	&	MathSAT	&	Z3	&	MathSAT &	Z3	&	MathSAT & Z3 & MathSAT \\
		\midrule
		\multirow{5}{1.5cm}{\centering$M^{+}$: $u_1=0$, $u_2\in [-1,1]$,
		template unknowns = 9}	&
		$0.4$	&	138.48	& 45.51		&	$ 8.16 \pm 0.00 $	&	$ 4.92 \pm 0.00 $	 &	$ 0.19 \pm 0.04 $ & $
		0.35 \pm 0.07 $ & 22192-18-22112-17 & 22-3-0-13\\
		& $0.3$	&	310.30	& 45.62		&	$ 6.04 \pm 3.01 $	&	$ 2.67 \pm 3.19 $ &	$ 0.19 \pm 0.04 $	& $ 0.35
		\pm 0.07 $	 & 52400-20-52320-20 & 22-4-0-14\\
		& $0.2$	&	TO	&	60.49	&	-	&	$ 2.67 \pm 3.19 $ &	-	& $ 0.30
		\pm 0.10 $	 & - & 30-5-0-19\\
		& $0.15$	&	TO	&	90.90	&	-	&	$ 1.52 \pm 2.27 $ &	-	&$
		0.30 \pm 0.10 $	 & - & 62-7-0-37\\
		& $0.125$	&	TO	&	TO	&	-	& -	 &	-	& -	 & - & -\\
		\midrule	
		\multirow{5}{1.5cm}{$M^{\times}$: $u_1\in [-1,1]$, $u_2=1$,
		template unknowns = 9}	& $0.4$	&
		57.47	&	75.38	&	$ 11.54 \pm 0.73 $	& $ 13.69 \pm 0.67 $	 &
		$ 7.11 \pm 0.00 $	& $ 11.75 \pm 0.00 $	 & 22-3-8-8 & 38-4-0-22 \\
		& $0.3$	&	59.14	&	89.75	&	$ 11.88 \pm 0.78 $	& $ 10.90 \pm
		4.77 $	 &	$ 7.11 \pm 0.00 $	& $ 11.75 \pm 0.00 $	 & 22-4-8-8 & 62-7-0-37\\
		& $0.2$	&	73.87	&	117.58	&	$ 11.90 \pm 0.64 $	& $ 8.87 \pm
		5.45 $	 &	$ 7.11 \pm 0.00 $	& $ 8.68 \pm 3.30 $ & 34-5-13-10	 & 154-14-0-90\\
		& $0.15$	&	83.27	&	136.08	&	$ 9.72 \pm 4.92 $	& $ 9.17 \pm
		5.14 $	 &	$ 6.76 \pm 0.49 $	& $ 8.29 \pm 2.63 $ & 46-7-18-14	 & 218-21-0-137\\
		& $0.125$	&	83.32	&	153.65	&	$ 9.72 \pm 4.92 $	&$ 8.07 \pm
		5.16 $	 &	$ 6.76 \pm 0.49 $	& $ 8.22 \pm 2.41 $ & 46-7-18-14	 & 298-31-0-209\\
		\midrule	
		\multirow{5}{1.5cm}{$M^{+,\times}$: $u_1=0$, $u_2\in [-1,1]$,
		template unknowns = 11 }	& $0.4$	&	214.01	&	76.56	&	$ 1.60
		\pm 0.75 $	& $ 0.61 \pm 0.32 $	 &	 $ 0.23 \pm 0.05 $	& $ 0.57 \pm
		0.47 $	 & 1568-28-1183-325 & 86-10-0-52\\
		& $0.3$	&	TO	&	91.74	&	-	& $ 0.64 \pm 0.29 $	 &	-	& $ 0.61
		\pm 0.39 $	 & - & 134-15-0-81\\
		& $0.2$	&	TO	&	122.26	&	-	&  $ 1.37 \pm 1.15 $	 &	-	& $
		0.72 \pm 0.40 $	 & - & 270-27-0-181\\
		& $0.15$	&	TO	&	144.18	&	-	&	$ 2.47 \pm 2.20 $ &	-	& $
		0.92 \pm 0.52 $	 & - & 360-42-0-256\\
		& $0.125$	&	TO	&	177.40	&	-	&	$ 3.43 \pm 3.44 $ &	-	&
		$ 2.03 \pm 2.06 $ & - & 530-65-0-403\\
		\bottomrule	
	\end{tabular}
	}%
	\vspace{-1em}
\end{table}

\end{document}